\documentclass[accepted]{uai2026} 
                        
\usepackage[american]{babel}

\usepackage{natbib} 
\usepackage{mathtools} 
\usepackage{amsfonts}
\usepackage{booktabs} 
\usepackage{tikz} 
\usepackage{amsthm}
\usepackage{delimset}
\usepackage{subcaption}

\makeatletter

\renewcommand\p@subfigure{}
\makeatother 

\DeclareMathOperator*{\argmax}{arg\,max}
\DeclareMathOperator*{\argmin}{arg\,min}
\newtheorem{theorem}{Theorem}
\newtheorem{lemma}{Lemma}

\newtheorem{proposition}{Proposition}
\newtheorem{definition}{Definition}

\newtheorem{condition}{Condition}

\newtheorem{algo}{Algorithm}

\title{Constrained Weighted Bayesian Bootstrap}

\author[1]{Sam~Rosen}
\author[2]{Jason~Xu}
\affil[1]{%
    Department of Statistical Science\\
    Duke University 
}
\affil[2]{%
    Department of Biostatistics\\
    University of California, Los Angeles 
}
  
\begin{document}
\maketitle

\begin{abstract}
    We prove the weighted Bayesian bootstrap, a method for approximate sampling of a posterior distribution, can be extended to sample from general constrained posterior distributions under mild assumptions. The method entails a simple algorithm that can take advantage of fast tools from convex optimization. Under regularity conditions, we show the asymptotic distribution of samples from the constrained weighted Bayesian bootstrap has a covariance matching  the restricted maximum likelihood estimator, an efficient estimator. We assess the method empirically on a variety of constrained Bayesian problems, demonstrating broad applicability of the method as well as advantages over existing peer methods. The constrained weighted Bayesian bootstrap quickly samples from constrained posteriors,  providing adequate uncertainty quantification for problems typically solved via optimization methods designed to deliver only a point estimate. As a case study, using constraints required in European-style option prices, uncertainty estimates of an option pricing surface are derived with constrained weighted Bayesian bootstrap.
\end{abstract}

\section{Introduction}\label{sec:intro}

The canonical formulation of statistical learning tasks from an optimization lens seeks to maximize a measure of fit subject to a constraint or regularization term imposing desired structure. Accordingly, there is a vast literature with classical theory and efficient algorithms for solving constrained problems. Generally, the solutions delivered by these methods provide a point estimate of the parameters of interest. Constructing confidence intervals can quickly become complicated by common constraints such as sparsity, order or shape restrictions. Uncertainty estimates could then necessitate post-selection inference procedures, deviate from standard asymptotic theory, or otherwise require case-by-case treatment depending on the constraint structure \citep{lee2016exact,panigrahi2018scalable}.

Alternatively, Bayesian methods yield straightforward uncertainty quantification via the posterior distribution, but the computation and sampling procedure itself is often complicated by the presence of the constraints.
When fitting a probabilistic model to data, the underlying parameters may require restriction to a subspace, $\tilde \Theta \subset \Theta$, of a larger parameter space. Exact sampling of a posterior distribution with support restricted to a general set $\tilde \Theta$ becomes difficult, even when the unconstrained posterior is well-defined. Popular generic sampling methods such as Hamiltonian Monte Carlo (HMC) cannot be easily used when imposing the constraint results in a lack of smoothness, and  variants for special cases such as restrictions along manifolds entail higher computational overhead \citep{girolami2011riemann}.  Sampling from the unconstrained posterior and then discarding samples outside the constraint set \textit{post hoc} will give valid samples, but can be extremely inefficient; in many cases $\tilde \Theta$ may be a low-probability or even zero-measure subspace of the posterior \citep{xu2023bayesian}.

Existing literature for sampling a posterior distribution with restricted support varies in the underlying mechanism. As a reasonable heuristic, transformation-based approaches sample from unconstrained posteriors and transform the samples to the constrained space \textit{ex post facto}, such as with an orthogonal projection \citep{astfalckPosteriorProjectionInference2024}, or an oblique projection determined from the posterior \citep{dunsonBayesianInferenceOrderConstrained2003, everinkBayesianInferenceProjected2023}. Orthogonal projections may ignore the geometry of the posterior and in cases where the posterior covariance or Fisher information must be estimated, oblique projections can be sensitive to this uncertainty. Specific prior forms encourage $\theta$ samples to be in $\tilde \Theta$ by heavily down-weighting the prior outside the constraint set with a hyperparameter controlling the amount of relaxation \citep{duanBayesianConstraintRelaxation2020, presmanDistancetoSetPriorsConstrained2023}; however, this may lead to samples that are close to, but not exactly in the constraint set. Although inexact, these methods are designed to interface nicely with gradient-based sampling methods. Some specialized priors take advantage of the robustness of the proximal mapping used in constrained optimization \citep{zhouProximalMCMCBayesian2024, xuBayesianInferenceUsing2024}. Application of this principled approach is limited to constraint sets with the ability to repeatedly calculate proximal mappings quickly. Altogether, constrained sampling with these methods has various trade-offs, and many share a connection to an idea from convex optimization.

In this work, we expand upon the weighted likelihood bootstrap first introduced in \citet{newtonApproximateBayesianInference1994}, which transforms a posterior sampling problem into an optimization problem (see Section \ref{sec:methods} for details). The weighted likelihood bootstrap has been generalized to a variety of contexts with success such as multimodal posteriors \citep{fongScalableNonparametricSampling2019}, nonparametric learning \citep{lyddonGeneralBayesianUpdating2019}, and latent variable models \citep{hanStatisticalInferenceMeanField2019}. In more recent works, where the prior form is explicitly considered, it is referred to as the weighted Bayesian bootstrap \citep{newtonWeightedBayesianBootstrap2021}. The weighted Bayesian bootstrap is effective with lasso-based estimators for both sparse normal mean and high-dimensional regression problems with compelling theoretical properties such as concentration towards the true underlying parameters \citep{ng_random_2022}. The work of \citet{nieBayesianBootstrapSpikeandSlab2023} produces samples of the spike-and-slab lasso posterior \citep{rockova-spike-and-slab_2018} using the weighted Bayesian bootstrap while proving a minimax optimal rate of concentration. Similar ideas lead to an expansion to handle binary responses in \citet{menacher-bayesian-2024}. 

Given the weighted likelihood/Bayesian bootstrap replaces standard posterior sampling using optimization steps, we propose a simple adaptation to enable sampling from posteriors with constrained support. This is especially salient as constraints can significantly hinder sampling-based posterior computation, while there is rich literature on constrained optimization. The weighted Bayesian bootstrap has shown robust efficacy both empirically and theoretically in a variety of domains, and this article will additionally show that many of these properties carry over to the constrained regime. Furthermore, the ease of implementation is shown for several examples, along with the difficulties of sampling in constrained spaces for other modern Bayesian posterior sampling methods\footnote{All code used to produce the experimental results and figures is found at \url{https://github.com/SamGRosen/CWBB}.}.

\section{Methods}\label{sec:methods}

We consider data generated from a probability density function $f_{\theta}$, with $\theta \in \Theta$ and prior beliefs encoded in $\pi(\theta)$. The standard form for the posterior distribution of $\theta$ given the underlying data follows
\begin{equation}
    \pi(\theta \mid x) \propto f_\theta(x) \pi(\theta). \label{eqn:standard-posterior}
\end{equation}
By restricting $\pi(\theta)$ to have positive support only on a set $\tilde \Theta \subset \Theta$, the \textit{constrained} posterior distribution can be written up to proportionality 
\begin{equation}
    \tilde \pi(\theta \mid x) \propto \pi(\theta \mid x) \times \mathbf 1(\theta \in \tilde \Theta), \label{eqn:constrained-posterior}
\end{equation}
where $\mathbf 1$ denotes the 0--1 indicator function.

As surveyed in the introduction, it is not straightforward to account for the indicator constraint in generality. To make progress, we revisit an idea introduced in \citet{newtonApproximateBayesianInference1994} that performs approximate sampling efficiently when the  \textit{maximum a posteriori} estimate of \eqref{eqn:standard-posterior} is available. They propose the weighted likelihood bootstrap, where samples are produced by first generating weights $w_{n} \sim n \times \text{Dirichlet}(1,\ldots,1)$, and then maximizing a weighted log-likelihood function,
\begin{equation}
    \argmax_{\theta \in \Theta} \sum_{i = 1}^n w_{n,i} \ell(x_i \mid \theta), \label{eqn:wlb-argmax-def}
\end{equation}
where $\ell(x_i \mid \theta) = \log f_{\theta}(x_i)$. These samples are shown to be first-order correct, sharing an asymptotic conditional distribution with the posterior distribution under a variety of priors. This may also be called the weighted Bayesian bootstrap when prior terms (also possibly weighted) are appended to \eqref{eqn:wlb-argmax-def} \citep{newtonWeightedBayesianBootstrap2021}.

\citet{newtonApproximateBayesianInference1994} posits repeated perturbations of the maximum likelihood equation on a single dataset create a path toward sampling. Motivated by the well developed literature and breadth of tools for constrained problems in the optimization literature, we devise the constrained weighted Bayesian bootstrap (CWBB), a variant of the weighted likelihood bootstrap designed to address \eqref{eqn:constrained-posterior}. To incorporate the prior constrained to $\tilde \Theta$, we add the log-prior term to \eqref{eqn:wlb-argmax-def} and write an equivalent form with an equality constraint
\begin{equation}
    \argmax_{\theta \in \Theta}\sum_{i=1}^n w_{n,i} \ell(x_i \mid \theta) + \log \pi(\theta), \quad h(\theta) = 0, \label{eqn:restricted-weighted-MAP}
\end{equation}
where $h\colon \mathbb R^{p} \mapsto \mathbb R^r$, $h(\theta) = 0$ if and only if $\theta \in \tilde \Theta$, and $h$ satisfies regularity conditions such as continuous second derivatives in a neighborhood of $\theta_0$. CWBB with this optimization problem is summarized in Algorithm \ref{algorithm:cwbb}. 

\begin{algo}\label{algorithm:cwbb}
Constrained weighted Bayesian bootstrap.
\vspace*{-8pt}
\begin{tabbing}
   \qquad \enspace For $t=1$ to number of desired samples \\
   \qquad \qquad Sample $w_n \sim n \times \text{Dirichlet}(1,\ldots,1)$. \\
   \qquad \qquad Set $\theta^{(t)}$ according to \eqref{eqn:restricted-weighted-MAP}. \\
\qquad \enspace Output $\theta^{(\cdot)}$
\end{tabbing}
\end{algo}

\paragraph{Illustrative example} Consider the Bayesian linear regression where coefficients are restricted to the unit sphere in $\mathbb R^p$ after a linear transformation:
\begin{equation}
\begin{split}
    y \mid X, \beta & \sim N(X \beta, \tau_{\epsilon}^{-1} I_n), \\
    \beta & \sim N(0, \tau_{\beta}^{-1}I_p)\times \mathbf 1({\|A \beta\|_2^2 = 1}), \label{eqn:example-posterior}
\end{split}
\end{equation}
where $A \in \mathbb R^{p \times p}$. Denoting the $i$th row of $X$ by $x_i$ and ignoring proportionality constants, we can write \eqref{eqn:restricted-weighted-MAP} as
\begin{equation}
    \argmin_{\|A \beta\|_2^2 = 1} \frac{\tau_{\epsilon}}{2}\sum_{i=1}^n  w_{n,i} \brk[c]{y_i - 
    x_i ^\top \beta}^2 + \frac{\tau_\beta}{2}\|\beta\|_2^2. \label{eqn:example-argmin}
\end{equation}
Let $\tilde y_i = w_{n,i}^{1/2} y_i$ and $\tilde X = \operatorname{diag}(w_{n})^{1/2} X$. The Lagrangian of \eqref{eqn:example-argmin} and resulting gradient is
\begin{align}
    \mathcal L(\beta, \lambda) & = \frac{\tau_{\epsilon}}{2} \| \tilde y - \tilde X \beta \|_2^2 + \frac{\tau_\beta}{2}\|\beta\|_2^2 + \lambda (\|A \beta\|_2^2 - 1), \nonumber \\
    \nabla_{\beta} \mathcal L(\beta, \lambda) & = -\tau_\epsilon \tilde X^\top (\tilde y - \tilde X \beta) + \tau_\beta \beta + 2\lambda A^\top A \beta. \label{eqn:example-lagrangian}
\end{align}
Setting \eqref{eqn:example-lagrangian} to zero for stationarity and solving for $\beta$ gives
\begin{equation}    
    \beta = \brk[c]{\tilde X^\top \tilde X + \tau_\epsilon^{-1}(\tau_\beta I_p + 2\lambda A^\top A)}^{-1} \tilde X^\top \tilde y. \label{eqn:example-multiplier-solve}
\end{equation}
Hence, the $p$-dimensional optimization problem is reduced to a one-dimensional root-finding problem for the Lagrange multiplier, $\lambda$, that satisfies both $\|A\beta\|_2^2 = 1$ and \eqref{eqn:example-multiplier-solve}. In the unconstrained case, computation in the weighted Bayesian bootstrap reduces to calculating ridge regression solutions for randomly weighted data. In the constrained case, CWBB operates similarly, but the regularization parameter automatically adapts to debias the solutions towards the constraint set across iterations. 

Figure \ref{fig:constrained-sample-examples} shows samples from several peer methods for sampling from the constrained posterior of \eqref{eqn:example-posterior} with $n = 100$, $p=2$ and $A = I_2$. Exact constrained samples concentrate on the unit circle around the true value $\beta_0$. Both methods which relax the constraint lead to separate issues: the distance-to-set prior of \cite{presmanDistancetoSetPriorsConstrained2023} has samples very close to the constraint set, but still noticeably far despite a high regularization value of $\rho = 5000$; the constraint relaxation from \cite{duanBayesianConstraintRelaxation2020} generates samples far from regions of high posterior density induced by a local mode from the nonconvex constraint set. Ignoring the geometry of the posterior, orthogonally projected samples may project to low density regions of the posterior. This is ameliorated in the obliquely projected samples since the posterior covariance is known. Meanwhile, CWBB respects the constraint set, and we see that its samples most closely resemble the exact posterior density.

\begin{figure}
\centering
\includegraphics[width=0.48 \textwidth]{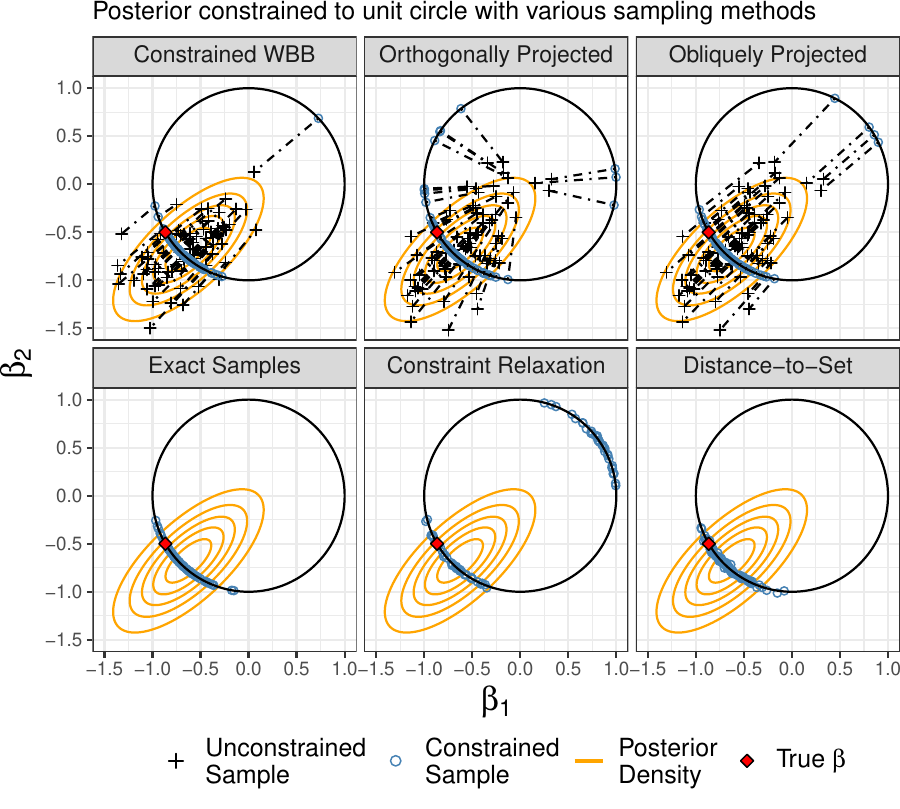}
\caption{Graphs showing 75 samples of $\beta$ restricted to the unit circle under six constrained sampling methods.}
\label{fig:constrained-sample-examples}
\end{figure}

\section{Theory}\label{sec:theory}

We establish several favorable properties of the posterior samples under CWBB. Under standard regularity conditions detailed in the Supplementary Material, the solution to \eqref{eqn:restricted-weighted-MAP} will exist and be conditionally consistent for $\theta_0$, along almost every sample path asymptotically. The regularity conditions are a natural extension to constraints for those used in analysis of the weighted likelihood bootstrap \citep{newtonApproximateBayesianInference1994}. In addition, we establish CWBB has a clear asymptotic distribution that gives valid frequentist confidence intervals. The restricted maximum likelihood estimator from \citet{aitchisonMaximumLikelihoodEstimationParameters1958}, 
\begin{equation}
    \theta_n^* = \argmax_{\theta \in \Theta} \sum_{i=1}^n \log f_{\theta}(x_i), \quad h(\theta) = 0, \label{eqn:restricted-mle}
\end{equation} 
is asymptotically efficient, achieving the Cram\'er-Rao lower bound under parametric constraints \citep{stoicaCramerRaoBoundParametric1998}. We establish that CWBB samples converge in distribution to a multivariate normal centered on this estimator with a matching covariance expression. To this end, we assume that $h$ has continuous bounded second derivatives in a neighborhood of $\theta_0$, and full-rank Jacobian at  $\theta_0$. We note that these conditions do not preclude $h(\theta)$ taking the form of an inequality constraint, but we focus our treatment on equality constraints without loss of generality. Inactive inequality constraints have no effect on the asymptotic distribution while active inequality constraints can be subsumed into the equality constraints.

\begin{theorem} \label{thm:one}
    Let $X_{1:n}$ be independent and identically distributed observations with underlying density function $f_{\theta_0}$. Let the Lagrangian of \eqref{eqn:restricted-weighted-MAP} be 
    \begin{equation*}
        \begin{split}
            \mathcal L(\theta, \lambda) = \sum_{i=1}^n w_{n,i} \ell(x_i \mid \theta) + \log \pi(\theta) + \lambda^\top h(\theta).
        \end{split}
    \end{equation*}
    Under regularity conditions, for all $\delta > 0$, there exists a sequence $\{\check \theta_n, \check \lambda_n\}$ such that
    \begin{align*}
        {\mathrm{pr}}\brk[c]{\nabla \mathcal L(\check \theta_n, \check \lambda_n) = 0 \mid X_{1:n}} & \to 1 \quad a.s.[X_{1:\infty}], \\
        {\mathrm{pr}}\brk[p]{\|\check \theta_n - \theta_0\|_2 < \delta \mid X_{1:n}} & \to 1 \quad a.s.[X_{1:\infty}].
     \end{align*}
\end{theorem}

\begin{theorem} \label{thm:two}
    Let $X_{1:n}$ be independent and identically distributed observations with underlying density function $f_{\theta_0}$. Let $\hat \theta_n$ be a strongly consistent estimator for $\theta_0$, satisfying $h(\hat \theta_n) = 0$ and
    \begin{equation}
        \bigg\| \frac{{I - U(\hat \theta_n)}}{n^{1/2}}  \sum_{i=1}^n \nabla_\theta \log f_{\hat \theta_n}(X_i) \bigg\|_2 \to 0\ a.s.[X_{1:\infty}] \label{eqn:strongly-consistent-requirement}
    \end{equation}
    where $D_h(\theta)$ is the Jacobian at $\theta$ for $h$, 
    \begin{align*}
        J(\theta) & = \mathbb E_{\theta_0}\brk[s]{\brk[c]{\nabla_{\theta} \log f_{\theta}(X)}\brk[c]{ \nabla_{\theta} \log f_{\theta}(X)}^\top}, \\
        U(\theta) & = D_h^\top(\theta) \brk[c]{D_h(\theta_0) J(\theta_0)^{-1} D_h^\top(\theta)}^{-1} D_h(\theta_0) J(\theta_0)^{-1}.
    \end{align*}
    Let $\check \theta_n$ be a sample from Algorithm \ref{algorithm:cwbb}. Then, under regularity conditions, for any Borel set $A \in \mathbb R^{p}$,
    \begin{equation*}
        {\mathrm{pr}}\brk[c]1{{n^{1/2}}(\check \theta_n - \hat \theta_n) \in A \mid X_{1:n} } \to {\mathrm{pr}}\brk{Z \in A}\ a.s.[X_{1:\infty}]
    \end{equation*}
    where $Z \sim N\brk[s]{0, J(\theta_0)^{-1} \brk[c]{I - U(\theta_0)}}$.
\end{theorem}
Note the restricted maximum likelihood estimator trivially satisfies \eqref{eqn:strongly-consistent-requirement} as demonstrated by the proof of Lemma 1 in \cite{aitchisonMaximumLikelihoodEstimationParameters1958}; see the proof of Lemma \ref{lemma:remove-multipliers} in the Supplementary Material for a similar proof.

Theorem \ref{thm:one} shows consistency of CWBB samples while Theorem \ref{thm:two} gives an asymptotic distribution for these samples. Theorem \ref{thm:two} additionally implies that the asymptotic covariance matrix is obtained by obliquely projecting the rows of the inverse Fisher information onto the kernel of $D_h(\theta_0)$, the Jacobian of the constraints at the true underlying parameter. This result is consistent with the asymptotic covariance of \eqref{eqn:restricted-mle}, giving the constrained form an explicit characterization for the uncertainty in approximate posterior samples. In addition, under these regularity conditions, the asymptotic distribution from Theorem \ref{thm:two} is invariant to the choice of $h$ used to describe $\tilde \Theta$ (see Lemma \ref{lemma:invariant-h} in the Supplementary Materials). Furthermore, we allow for  $\tilde \Theta$ to be a lower dimensional subspace of $\Theta$; this is not handled precisely in the analysis of constraint relaxation and is also lacking theoretical support in projecting the unconstrained posterior after sampling \citep{astfalckPosteriorProjectionInference2024}.

In particular, CWBB converges to a distribution with support only on the constraint set as the sample size grows. This is not the case for some other sampling methods proposed for \eqref{eqn:constrained-posterior}, such as constraint-relaxed posteriors of the form 
\begin{equation}
    \pi_{\rho}(\theta \mid x_{1:n}) \propto \exp\brk[c]3{\sum_{i=1}^n \log f_{\theta}(x_i) -\rho q(\theta)} \label{eqn:relaxed-posterior},
\end{equation}
where $q(\theta) = \operatorname{dist}^2(\theta, \tilde \Theta) / 2$ \citep{presmanDistancetoSetPriorsConstrained2023, zhouProximalMCMCBayesian2024} or more generally $q \geq 0$ such that $\theta \in \tilde \Theta \iff q(\theta) = 0$ \citep{duanBayesianConstraintRelaxation2020}. While these approaches to posterior inference coincide with the constrained problem as the penalty parameter $\rho \to \infty$, this relationship with the sample size is not uniform, so it is difficult to establish the analogous posterior consistency theory. Asymptotically, under a fixed penalty $\rho>0$, samples fall back to the standard Bernstein--von Mises theorem (see Chapter 10 of \citet{vaartAsymptoticStatistics1998}): if $q$ is finite in a neighborhood of $\theta_0$, then $n^{1/2}\theta \mid x_{1:n}$ converges to a distribution with variance $J(\theta_0)^{-1}$. This implies the constraints become less important to the posterior covariance and subsequent uncertainty estimation as more data is ingested, unless the relaxation adjusts to the sample size. In fact, using results on generalized posteriors from \citet{miller_asymptotic_2021}, Proposition \ref{proposition:penalized-posterior-distribution} shows that if the relaxation strength changes linearly with $n$, then the asymptotic covariance is full-rank, regardless of the dimension of $\tilde \Theta$.  

\begin{proposition}\label{proposition:penalized-posterior-distribution}
    Suppose $\Theta \subset \mathbb R^p$, $\rho > 0$, and $q(\theta)$ is three-times continuously differentiable in a closed ball around $\theta_0$. In addition, $f_\theta(x) = \exp\brk[c]{\theta^\top s(x) - \kappa{(\theta)}}$ is a full, regular, and identifiable member of the exponential family in natural form. Finally, the matrix $\Sigma~=~\brk[s]{-\mathbb E_{\theta_0}\brk[c]{\nabla_\theta^2 \log f_{\theta_0}(X)} + \rho \nabla^2_\theta {q(\theta_0)}}^{-1}$ is positive-definite.
    Then if $\theta \sim \pi_{n\rho}$ from \eqref{eqn:relaxed-posterior} there exists a strongly consistent estimator $\hat \theta_n$ such that 
    \begin{equation*}
        d_{\text{TV}}\brk[c]{n^{1/2}(\theta - \hat \theta_n), N(0, \Sigma) } \to 0,
    \end{equation*}
    almost surely, where $d_{\text{TV}}$ is the total variation distance.
\end{proposition}

For concreteness, Proposition \ref{proposition:penalized-posterior-distribution} implies the constraint-relaxed posterior for many exponential family distributions with $\rho' = n\rho$ will have an asymptotic covariance slightly informed by the constraints. As an example, if $q(\theta) = \operatorname{dist}^2(\theta, \tilde \Theta)/2$ where $\tilde \Theta = \brk[c]{\theta \colon P \theta = \theta}$ and $P$ is an orthogonal projection matrix, then the asymptotic covariance is $\brk[c]{J(\theta_0) + \rho (I - P)}^{-1}$.

\section{Experiments} \label{section:experiments}
\subsection{Nonnegative Order-Constrained Parameters in Regression} \label{section:nonnegative-ordered}

\begin{figure*}[t]
\begin{subfigure}{0.48\textwidth}
\includegraphics[width=0.99\columnwidth]{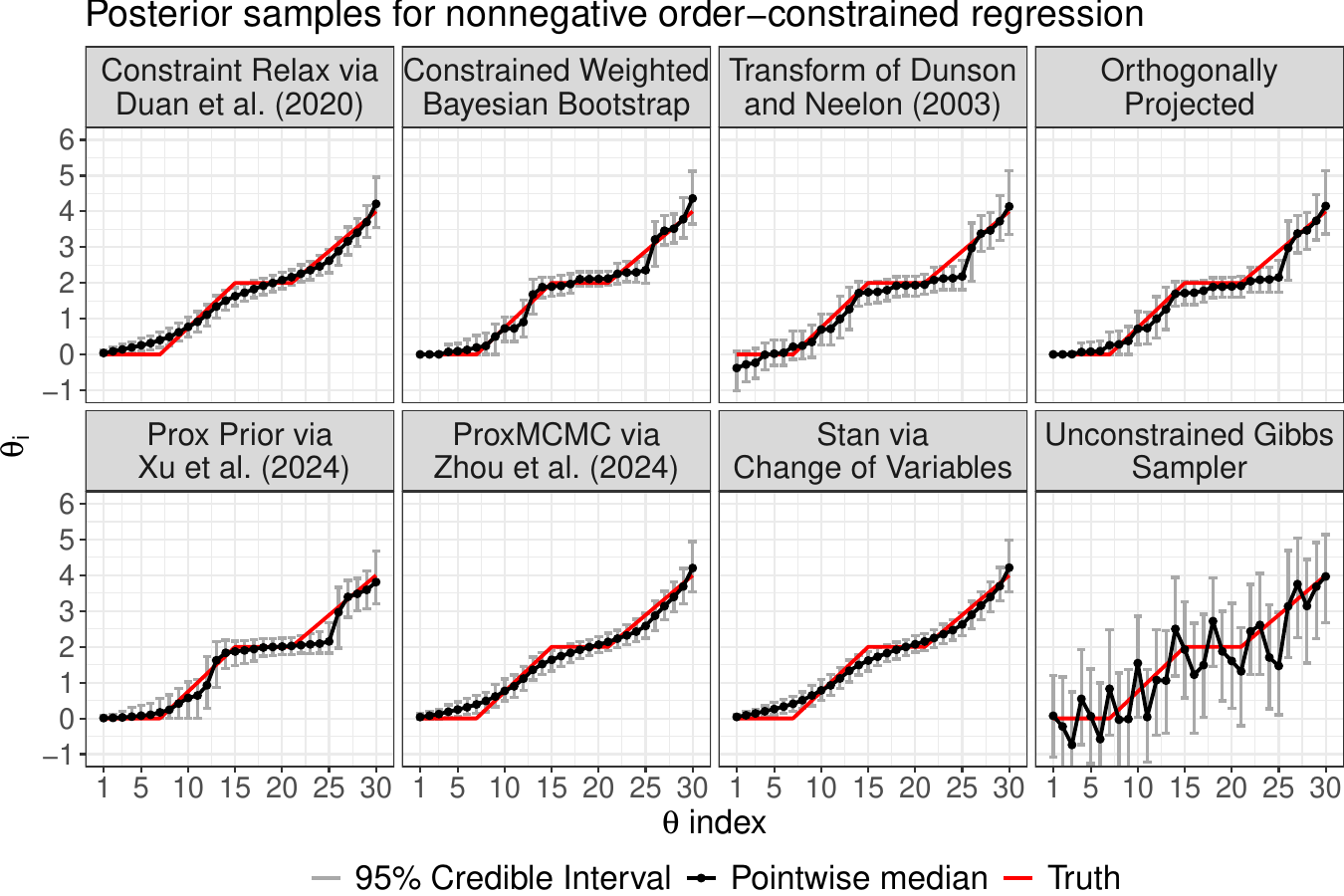}
\caption{A single trial for sampling from \eqref{eqn:posterior-nonnegative-ordered} under several peer methods with $n = 100$.}
\label{fig:ordered-samples}
\end{subfigure} \hspace{12pt}
\begin{subfigure}{0.49\textwidth}
\includegraphics[width=0.99\columnwidth]{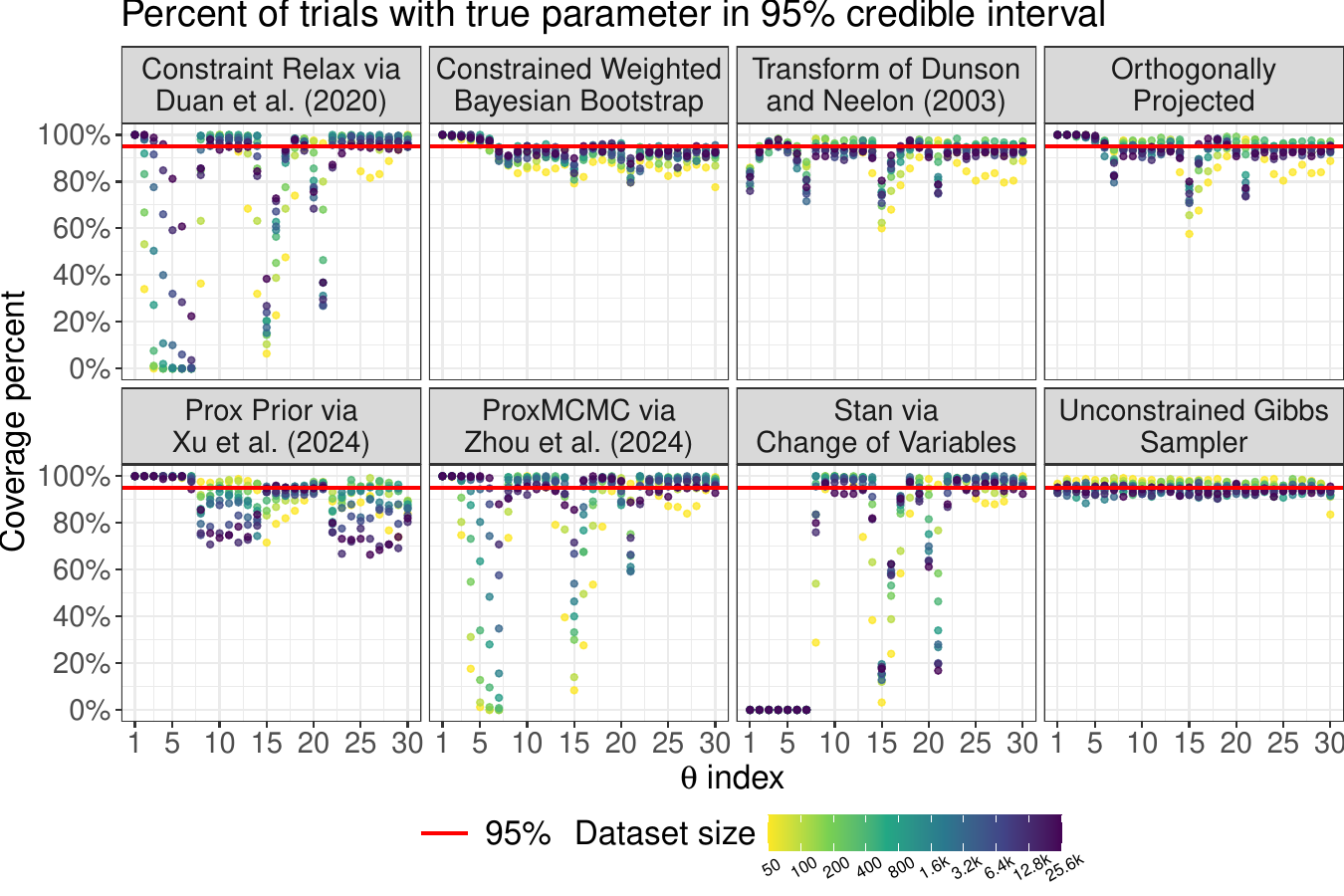} 
\caption{Percent of trials with true $\theta_i$ in the 95\% posterior credible intervals, for each $\theta_i$, dataset size, and sampling method.}
\label{fig:ordered-coverage}
\end{subfigure}
\end{figure*}

In \citet{dunsonBayesianInferenceOrderConstrained2003}, the authors describe a Bayesian linear model under the restriction where the learned coefficients are both nonnegative and nondecreasing according to underlying indices. This model is motivated by a biomedical study with ordered categorical predictors, e.g. increasing intensities of treatment regimes. This setting can be modeled with the Bayesian system
\begin{equation}
\begin{split}
    & y_i \mid \theta, \tau, x_i \sim N(x_i^\top \theta, \tau^{-1}), \\
    & \theta \sim N(\theta_0, \Sigma_0), \quad \tau \sim \text{Gamma}(a_0, b_0), \label{eqn:posterior-nonnegative-ordered}
\end{split}
\end{equation}
where $\theta$ is restricted to the convex space
\begin{equation*}
    \mathcal C = \brk[c]{\theta \in \mathbb R^p \colon 0 \leq \theta_1 \leq \theta_2 \leq \ldots \leq \theta_p},
\end{equation*}
and $x_i$ are covariates. Equivalently, there exists an invertible $D \in \mathbb R^{n\times n}$ such that $\theta \in \mathcal C \iff D\theta \geq 0$. Using arguments from \citet{gelfandBayesianAnalysisConstrained1992} surrounding constrained Gibbs sampling, \citet{dunsonBayesianInferenceOrderConstrained2003} generalize a method from \citet{hwangConfidenceIntervalEstimation1994} to transform unconstrained posterior samples via
\begin{equation}
    \hat \theta_i = \min_{i \leq t} \max_{s \leq i} \frac{\mathbf 1_{t - s + 1}^\top (\Sigma_{[s:t]})^{-1} \theta_{[s:t]}}{\mathbf 1_{t - s + 1}^\top (\Sigma_{[s:t]})^{-1} \mathbf 1_{t-s+1}}, \label{eqn:dunson-neelon}
\end{equation}
where $\Sigma$ is the posterior covariance and $[s:t]$ is indexing for the $s$th to $t$th elements or submatrix.
This does not guarantee transformed parameters are ordered or nonnegative, although the authors note this tends to be a nonissue in practice. The transformation requires solving a quadratic number of linear systems, accruing $O(p^4)$ operations to ensure that the parameter vector obeys the constraint.

CWBB guarantees samples satisfying the constraints much more efficiently. Each iteration samples  $w_n~\sim~n~\times~\text{Dirichlet}(1, \ldots, 1)$ and solves
\begin{equation*}
\begin{split}
    & \argmin_{\theta \in \mathcal C, \tau > 0}\ -(n/2 + a_0 - 1) \log \tau + b_0 \tau \\
    & \ + \frac{\tau}{2} \sum_{i=1}^n w_{n, i} (y_i - x_i^\top \theta)^2 + \frac{1}{2}(\theta - \theta_0)^\top \Sigma_0^{-1} (\theta - \theta_0). \\
\end{split}
\end{equation*}
Solutions to the optimization subproblem are delivered by coordinate descent on the $\theta$ and $\tau$ blocks. For fixed $\tau$, a convex quadratic program with linear constraints emerges, giving a multitude of methods to optimize for $\theta$; for fixed $\theta$, the optimal $\tau$ is available in closed form.

We implement the above to sample from the posterior of \eqref{eqn:posterior-nonnegative-ordered} constrained to $\mathcal C$ using CWBB. Simulation parameters were in a regime with $p = 30$ and true parameters $\tau = 1/25$ and $\theta$, a sequence of nondecreasing values with two constant contiguous subsequences (see Fig.~\ref{fig:ordered-samples} for an illustration). We sample each covariate vector, $x_i$, from a multivariate normal with mean zero and covariance matrix with unit diagonal, and $0.6, 0.3, 0.1$ on the main off-diagonals. For prior hyperparameters, $\theta_0$ is the zero vector, $\Sigma_0 = 2I_p$ and $a_0, b_0 = 1$.

In addition, six other constrained sampling methods are used to sample this posterior, which are listed in Fig.~\ref{fig:ordered-samples}. Note that the Moreau-Yosida envelope of an orthogonal projection is the distance-to-set function, so sampling via ProxMCMC is equivalent to the distance-to-set prior of \citet{presmanDistancetoSetPriorsConstrained2023} in the case of priors with fixed constrained support. We also sample using the \texttt{positive\_ordered} change of variables implemented in Stan \citep{Stan}, which leads to samples strictly in the interior of $\mathcal C$. This leads to difficulties in coverage as the true $\theta_0$ is on the boundary of $\mathcal C$, a setting supported by CWBB according to the theory in Section \ref{sec:theory}.

With these implementations, 2500 total trials across ten different dataset sizes are simulated to determine coverage abilities of the constrained samples from a variety of methods. Figure~\ref{fig:ordered-samples} displays an example set of credible intervals for a single trial with $n = 100$. For each dataset and method, 95\% credible intervals for each $\theta_i$ are constructed. The coverage percent is then the percent of trials which contain the true $\theta_i$ inside the respective credible interval. Figure~\ref{fig:ordered-coverage} displays the coverage percents for all $\theta_i$ over a variety of dataset sizes for several constrained posterior sampling methods and an unconstrained Gibbs sampler as a baseline. Of the constrained sampling methods, CWBB has coverage most resembling that of the unconstrained Gibbs sampler, which does not respect the constraints and has much wider intervals, as seen in Fig.~\ref{fig:ordered-samples}. Many methods struggle with coverage where the true $\theta_i$ is flat or changes slope, with CWBB suffering the least from these difficulties. In the Supplementary Material, Table \ref{tab:nonneg-order-runtimes} has the total runtime to complete all 2500 trials, showing CWBB is orders of magnitude faster than several of the competing methods. In addition, see Fig.~\ref{fig:minimum-coverage} to compare the minimum coverage percent of all $\theta_i$ for the methods, showing CWBB overall has the best worst-case performance as compared to the unconstrained  Gibbs baseline. 

\subsection{Sparse Precision Matrix Estimation}\label{section:bayes-gl}

Assuming data $x_i$ follow a multivariate normal distribution with an unknown sparse precision matrix, the maximum likelihood estimate of the precision matrix is its sample value. Sparsity is a useful characteristic of precision matrices allowing for explicit calculation of conditional dependence among the variables of $x_i$. To induce sparsity in the maximum likelihood estimate, the graphical lasso of \citet{friedmanSparseInverseCovariance2008} optimizes an $\ell_1$ penalty added to the likelihood,
\begin{equation*}
    \log \det \Sigma^{-1} - \operatorname{tr}( S \Sigma^{-1}) - 
    \rho \sum_{i \leq j} |\Sigma_{ij}^{-1}|,
\end{equation*}
where $S$ is the empirical covariance matrix. When a point estimate is desired, the graphical lasso is an effective tool, but uncertainty quantification around classification of zero entries can improve inference. Bayesian methods quantify uncertainty for precision matrix estimation with a variety of sparsity-inducing priors. Given samples from a posterior, a probability of edge inclusion (nonzero entry of precision matrix) is produced from the empirical distribution. The Bayesian graphical lasso of \citet{wangBayesianGraphicalLasso2012} uses a formulation similar to \eqref{eqn:bayesian-gl} where \textit{maximum a posteriori} estimates are sparse and inclusion probabilities are estimated via partial correlations from the posterior samples. The setup of the Bayesian graphical lasso allows a prior over the possible regularization parameters, giving a smaller need to tune. \citet{mohammadiBayesianStructureLearning2015a} uses the G-Wishart prior for more explicit inference of edge inclusion by directly sampling over the many possible graphical representations of variable conditional dependence. Posterior sampling is done via a birth-death process. \citet{wangScalingItStochastic2015} also considers a discrete approach via a spike-and-slab prior for all pairs of variables with sampling done using a block Gibbs sampler. ProxMCMC \citep{zhouProximalMCMCBayesian2024} features sparse precision matrix estimation as a canonical example, where samples are constrained to an $\ell_1$ ball after specifying a prior distribution over potential radii.

To sample from a posterior of sparse precision matrices, we apply CWBB by viewing the graphical lasso as the \textit{maximum a posteriori} estimate of the following Bayesian system:
\begin{equation}
\begin{split}
    x_i \mid \Sigma & \sim N(0, \Sigma), \\
    \pi(\Sigma^{-1}) & \propto \mathbf 1(\Sigma^{-1} \in \mathbb S^{p \times p}_+) {\prod_{i \leq j} \exp(-\rho |\Sigma^{-1}_{ij}|)}, \label{eqn:bayesian-gl}
\end{split}
\end{equation}
also used in \citet{wangBayesianGraphicalLasso2012}. The prior promotes sparsity by positing that entries of the precision matrix follow a Laplace distribution, but together they must also lie in the positive-definite cone, $\Sigma^{-1} \in \mathbb S_+^{p \times p}$, or equivalently $ \lambda_{\min} (\Sigma^{-1}) > 0$.  Hence, CWBB involves simply repeatedly sampling $w_n$ and solving
\begin{equation}
\begin{split}
    & \argmin_{\Sigma^{-1}\in \mathbb S^{p \times p}_+} -\frac{n}{2} \log \det \Sigma^{-1} \\
    & \quad + \frac{n}{2} \operatorname{tr}\brk[c]{\Sigma^{-1} X^\top \operatorname{diag}(w_n) X} + \sum_{i \leq j} q_\rho(\Sigma_{ij}^{-1}), \label{eqn:wbb-gl}
\end{split}
\end{equation}
where $q_\rho(\Sigma_{ij}^{-1}) = \rho |\Sigma_{ij}^{-1}|$. This can be optimized efficiently by the methods in \citet{friedmanSparseInverseCovariance2008}. After many samples from CWBB, we calculate the posterior credible intervals for each entry. 

The generality of CWBB allows uncertainty quantification for other penalties that induce sparsity. For example, the smoothly clipped absolute deviation \citep{fanVariableSelectionNonconcave2001} and minimax concave penalty \citep{zhangNearlyUnbiasedVariable2010} have desirable properties in problems requiring sparsity,
\begin{align}
   q_{\rho,SCAD}(z) & = \begin{cases}
       \rho |z| & |z| \leq \rho, \\
       \frac{2a\rho|z|-z^2-\rho^2}{2(a-1)} & \rho \leq |z| \leq a \rho, \\
       \frac{\rho^2(a+1)}{2} & |z| \geq a\rho,
   \end{cases}\label{eqn:scad} \\
   q_{\rho,MCP}(z) & = \begin{cases}
       \rho |z| - \frac{z^2}{2a} & |z| \leq a \rho, \\
       \frac{1}{2}a \rho^2 & |z| > a\rho.
   \end{cases} \label{eqn:mcp} 
\end{align}
By replacing the penalty $q$ in \eqref{eqn:wbb-gl} with \eqref{eqn:scad}, we can now utilize SCAD toward Bayesian sparse precision estimation efficiently; the point estimate was shown to be effective in \citet{fanNetworkExplorationAdaptive2009}. We also compare with the minimax concave penalty \eqref{eqn:mcp} from \citet{zhangNearlyUnbiasedVariable2010}, which exhibits similar performance and illustrates the flexibility of CWBB. For the analogous problem of sparse covariance estimation, CWBB also fits with principled optimization routines \citep{bien2011sparse, xu2022proximal}.

\begin{figure}
\centering
\includegraphics[width=.46\textwidth]{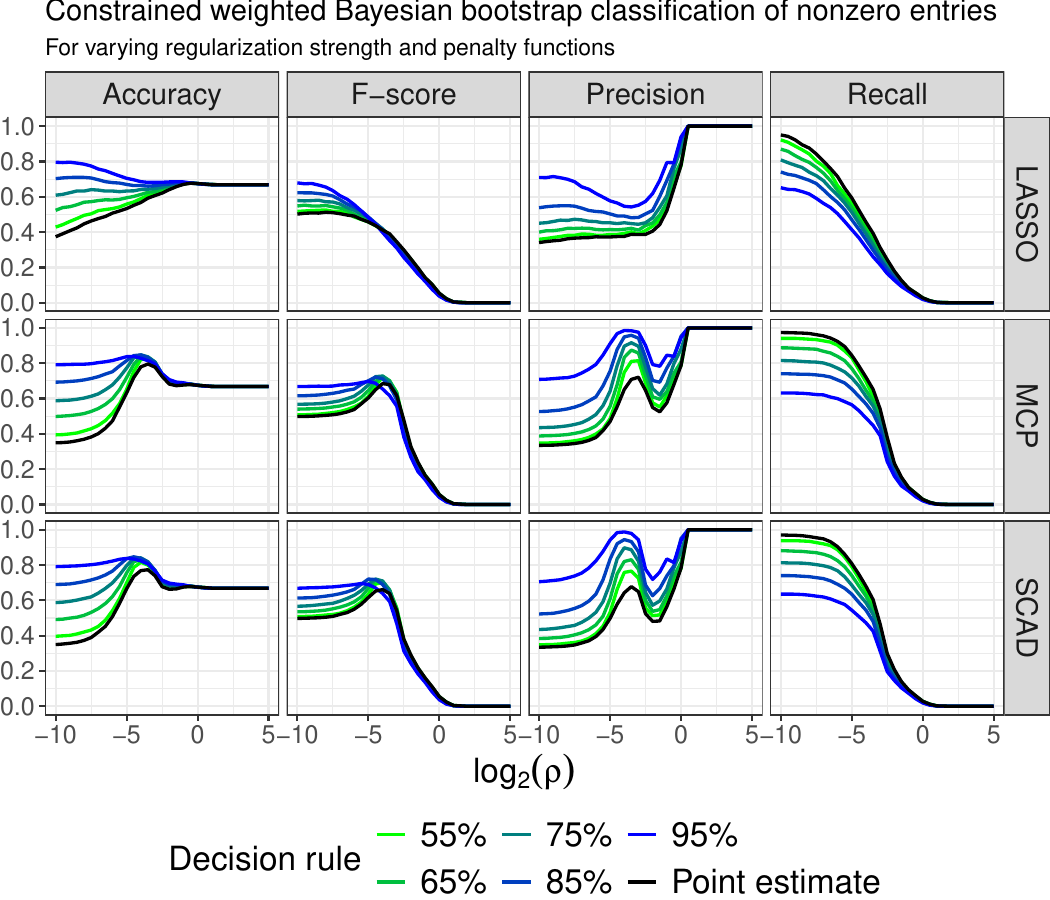}
\caption{Graph showing binary classification performance of nonzero entries of a sparse precision matrix with CWBB. The black line represents classifying each entry according to the point estimate at $\rho$.}
\label{fig:graphical-lasso-classification}
\end{figure}

The posterior credible intervals can guard against the well-documented issue of false selection with $\ell_1$ penalties \citep{su2017false, lafitPartialCorrelationScreening2019}. Rather than use a single graphical lasso estimate to determine non-zero entries, we employ a decision rule that classifies an entry as nonzero if the majority of the credible interval is either strictly positive or negative. To motivate casting the problem as posterior inference, if a credible interval is wide and almost symmetric about zero, it is reasonable to classify the entry as zero compared to one with a similar point estimate but with a credible interval with mostly positive support. 

We perform a simulation study with $p = 100, n = 750$ and $\Sigma^{-1} = L^\top L$, where $L$ is triangular and off-diagonal entries are independent mixtures of a standard normal and a point mass at zero. The sparsity of the true $\Sigma^{-1}$ is about two-thirds. CWBB performs 250 samples from \eqref{eqn:wbb-gl} for each penalty function. We use $a=3.7$ for \eqref{eqn:scad} and $a=3$ for \eqref{eqn:mcp} as both were shown to be reasonable defaults in their respective works. To evaluate performance of each method, entries of the precision matrix are treated as a binary classification problem with nonzero entries considered as positives. The classification metrics considered for evaluation are accuracy, precision, recall and $F_1$-score.

For many values of regularization strength, each row of Fig.~\ref{fig:graphical-lasso-classification} demonstrates an increase in performance when classifying precision matrix entries based on the credible intervals derived from CWBB as opposed to a single point estimate for all penalties considered. This performance is shown for a variety of decision rules dependent on the proportion of posterior samples that are strictly either positive or negative. This decision rule evidently results in an increase in false negatives as seen in the decrease in recall. However, the gain in precision outpaces these faults by both $F_1$-score and classification accuracy. Furthermore, there is a wide range of regularization hyperparameters that have competitive performance with the posterior samples, showing less sensitivity to tuning $\rho$ for all possible penalty functions. Performance of CWBB is competitive with the more sophisticated Bayesian methods mentioned earlier, while involving fewer hyperparameters (see Table~\ref{tab:bayesian-sparse-precision} of the Supplementary Material). The extra machinery of these methods requires some delicate handling of the burn-in period, starting values, and tuning, while CWBB is a simple, easy-to-use approach.

\subsection{Option Pricing Surfaces}\label{section:option-prices}

For a fixed stock, call options may be bought at a variety of strike prices and expiry dates. European-style call options give the buyer the right to purchase 100 shares of a stock at time $t$ for price $s$. Let $C(t, s)$ be the price of this option. The pricing function $C$ follows restrictions which prevent various types of arbitrage. Here, we focus on a set of minimal basic restrictions that are readily justifiable:
\begin{equation}
\begin{split}
    C(t, \cdot) & \text{ is convex and decreasing  for fixed $t$}; \\
    C(\cdot, s) & \text{ is increasing for fixed $s$}. \label{eqn:option-shape-constraints}
\end{split}
\end{equation}
This posits that the underlying price of a call option must decrease as the strike price increases, otherwise a clear arbitrage opportunity is available. Convexity is derived from principles in option pricing theory outside the scope of this work \citep{ait-sahaliaNonparametricOptionPricing2003}. Because call options risk expiry before an advantageous price is reached for the buyer, their value or price increases with lifespan, reflected in the second constraint.

We consider an additional constraint available from market data, as option contracts come with a bid and an ask price. The market price is generally taken to be the midpoint of the bid-ask spread, but this may not be an ideal estimate of the underlying price, and may violate the natural restrictions described above. Assuming the bid-ask is rational and the true option price falls within this spread, we add the constraint
\begin{equation}
    BID(t, s) \leq C(t, s) \leq ASK(t, s). \label{eqn:option-bid-ask-constraints}
\end{equation}

When pricing options that are part of a spread, both covariance and constraints play important roles. As an example, consider buying a vertical call spread, i.e. a trader buys a call expiring at time $t$ for strike $s_1$, and sells a call expiring at the same time for strike $s_2$ with $s_2 > s_1$. The net price is $P = C(t, s_1) - C(t, s_2)$. Therefore, a probabilistic model for $P$ must have support only on nonnegative reals  to satisfy \eqref{eqn:option-shape-constraints}. In addition, uncertainty estimates are essential for portfolio optimization, and the uncertainty for $P$ is changed by the covariance between $C(t, s_1)$ and $C(t, s_2)$.

Similar to prior work \citep{ait-sahaliaNonparametricOptionPricing2003}, we have a simple model specification, complemented by the complexity of the constraints. In particular, we choose a constrained Bayesian normal means formulation with a likelihood and unconstrained prior of
\begin{equation}
    \begin{split}
        y \mid C & \sim N_n(C, \lambda^{-1} I), \qquad  C \sim N_n(0, \delta^{-1} \Sigma), \label{eqn:option-prices-bayesian-system}
    \end{split}
\end{equation}
where $n$ is the number of options available for trade, $y$ is the observed midpoint $(BID + ASK)/{2}$, $\lambda$ and $\delta$ are precision hyperparameters and $\Sigma$ is a positive-definite matrix such that each entry is the prior covariance between a pair $C(t_1, s_1), C(t_2, s_2)$. In the notation of \eqref{eqn:option-prices-bayesian-system}, $C$ is flattened as a vector in $\mathbb R^n$, with indexing by $t$ and $s$ suppressed in the notation. We model the midpoint price as a noisy observation of the underlying price.
The posterior distribution of \eqref{eqn:option-prices-bayesian-system} is multivariate normal with covariance $\Sigma' = \brk[c]{\lambda I + \delta \Sigma^{-1}}^{-1}$ and mean $\lambda \Sigma' y$. As mentioned before, covariance is essential, yet this posterior derives covariance of option prices solely from the choice of hyperparameters $\lambda, \delta$ and $\Sigma$, ignoring the context of the data. By factoring in the constraints \eqref{eqn:option-shape-constraints} and \eqref{eqn:option-bid-ask-constraints} we induce more viable estimates of the covariance of the option prices under pertinent assumptions. 
 
\begin{figure*}[t]
\begin{subfigure}{0.49\textwidth}
\includegraphics[width=0.99\columnwidth]{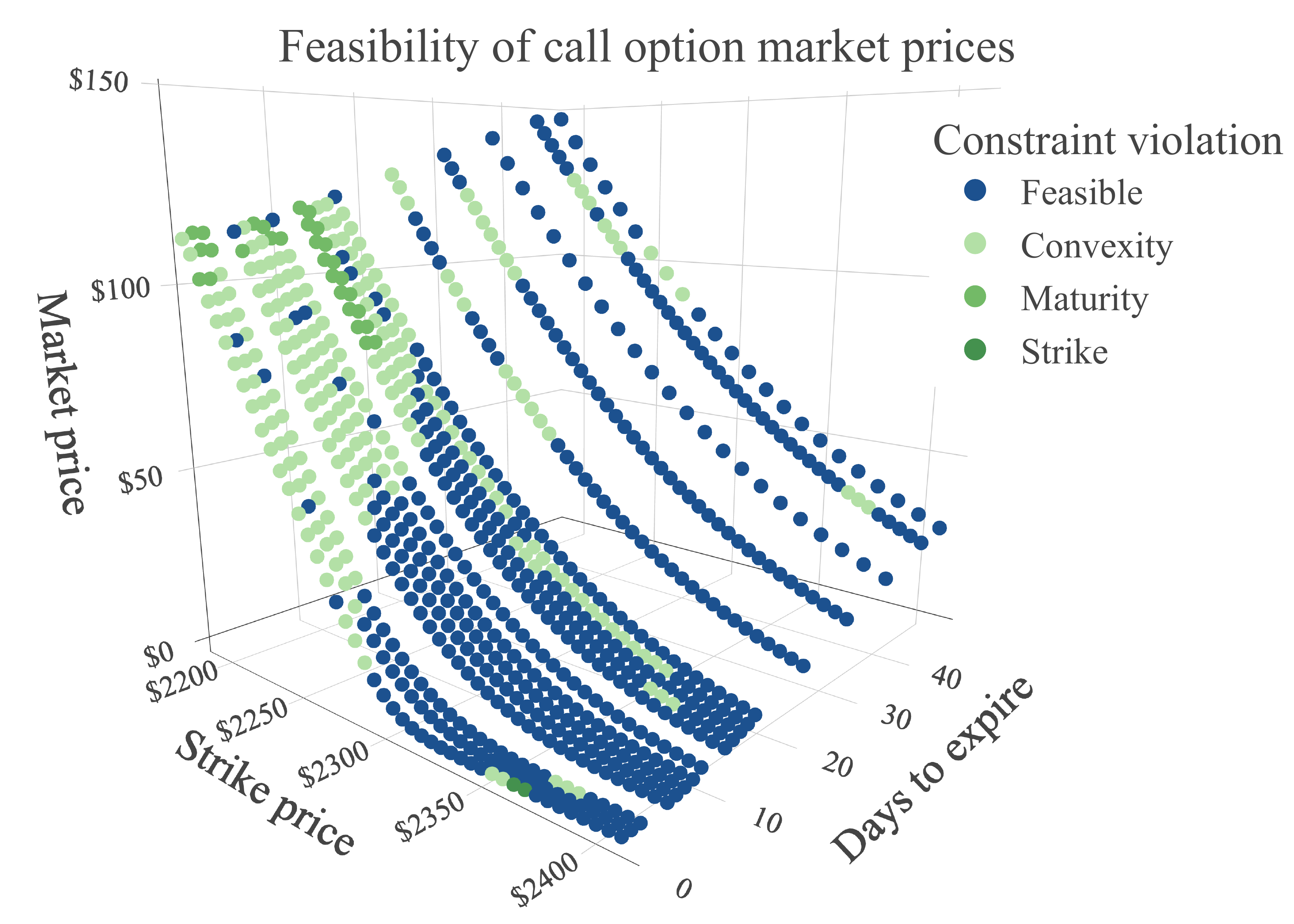}
\caption{For \texttt{RUT} call options, 3D scatter plots showing the midpoint price as it varies in strike price and days to expiration with points colored by constraint violations.}
\label{fig:option-prices-feasibility}
\end{subfigure} \hspace{14pt}
\begin{subfigure}{0.44\textwidth}
\includegraphics[width=0.99\columnwidth]{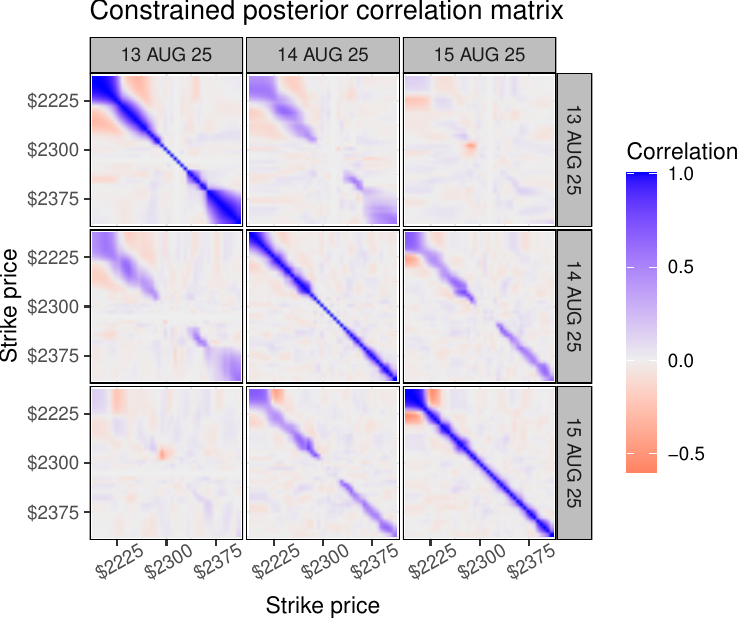} 
\caption{Correlation matrix for $C \mid y$ from CWBB for the first three times to expire in the dataset.}
\label{fig:option-prices-correlation-constrained}
\end{subfigure}
\end{figure*}

Iterations of CWBB solve 
\begin{equation}
    \argmin_{DC \geq b} \frac{\lambda}{2} \sum_{i=1}^n w_{n,i}\brk{C_i - {y_i}}^2 + \frac{\delta}{2} C^\top \Sigma^{-1} C, \label{eqn:options-cwbb}
\end{equation}
where $D$ is a matrix such that $DC \geq b$ implies feasibility. In particular, $D$ contains signed copies of a submatrix $D_1$ where $D_1 x \geq 0$ implies $x$ is nondecreasing and $D_2$ where $D_2 \geq 0$ implies a convex structure (see \citet{kimell_1TrendFiltering2009} for details). In addition, $D$ contains a positive and negative copy of the identity matrix to satisfy the bid-ask box constraints.

We perform experiments using market data from about 10 AM on August 13th, 2025 for the \texttt{RUT} European-style call options. Tradable options that expire within 60 days and have a strike price within 5\% of the underlying price of \$2300 are included in the dataset ($n = 766$). Each data point consists of a bid and an ask, and the midpoint of the bid-ask is used as $y_i$ in \eqref{eqn:options-cwbb}. The strike price and time to expire for each is used to calculate the prior covariance of $C$ while defining the constraints that resulting posterior samples must satisfy. In our exploratory data analysis and resulting inference, there are patterns corresponding to a typical classification in options: (Deep) In-the-money call options have a strike price (well) below the underlying price; near-the-money call options have a strike close to the underlying; (deep) out-of-the-money have a strike (well) above the underlying.

Figure~\ref{fig:option-prices-feasibility} shows which options have an observed midpoint price which violate some type of constraint, i.e. for option $j$ there exists some row $i$ in $D$ where $\operatorname{row}_i(D)^\top (BID + ASK) < 0$, and $D_{ij} \neq 0$. As an example, options expiring in 15 days with strike prices, \$2285, \$2290, and \$2295 have a midpoint price of \$48.00, \$45.00, and \$41.90, respectively. This violates convexity as the first derivative is not monotone increasing with respect to strike price, hence these three points are light green in Fig.~\ref{fig:option-prices-feasibility}. Most of the near-the-money options are not part of a constraint violation. On average, as the time to expire increases, with 15 days to expiry as an exception, the range of feasible options grows. Deep in-the-money call options are likely to violate convexity, the most common violation. Sequences of call options for several expiration dates violate maturity monotonicity.

We choose hyperparameters $\lambda = 1 / 3$ and $\delta = 1 / 35$. The hyperparameter $\Sigma$ is determined by an exponential kernel over the strike price--expiration date space, e.g. if $z_i = (s_i, t_i)$ is the strike and expiration of the $i$th option, then $\Sigma_{ij} = \exp\brk[c]{-\operatorname{dist}(z_i/\sigma_s, z_j/\sigma_t)}$, where $\sigma$ normalizes the coordinates. This makes the unconstrained prior of \eqref{eqn:option-prices-bayesian-system} informative, but the small value of $\delta$ compared to $\lambda$ leads the posterior covariance to have many tridiagonally dominated blocks, with over 93\% of entries having an absolute value less than $10^{-3}$. Using the \texttt{OSQP} solver of \citet{osqp}, 1000 samples of $C \in \mathbb R^{766}$ under the constrained posterior are found via CWBB. With these posterior samples under the constrained model, we can estimate ${\mathrm{cov}}\brk[c]{C(t_1, s_1), C(t_2, s_2) \mid y}$ for all expirations and strikes.

Figure~\ref{fig:option-prices-correlation-constrained} shows the empirical correlation matrix of the posterior parameters (option prices), for the first three times to expire in the snapshot, with blocks highlighting the pairs of expiration dates. Each tile corresponds to the sample correlation between the constrained posterior samples of $C(t_1, s_1)$ and $C(t_2, s_2)$. We see clear structure with blocks of high correlation for deep in-the-money call options as well as increasing amounts of correlation for deep out-of-the-money call options which expire on the same day. These blocks grow smaller as the time to expire increases, seen by the smaller correlation blocks for options expiring on the 15th of August, as opposed to the 13th. Furthermore, there is a small amount of correlation between options on adjacent days to expiry, but this correlation is weaker when examining options that expire two days apart. Near-the-money options also exhibit some weak correlation, although it does not have a block structure.

The constraints lead to uncertainty and point estimates related to the common classification of ``moneyness'' of the option. In the Supplementary Material, we display additional results on the constrained posterior, such as how it compares to the unconstrained problem, and how uncertainty relates to the spread of an option. In addition, we detail the limitations of several other constrained sampling methods for this case study. We find they are less viable, involving intractable runtimes, small effective sample sizes, or excessive constraint violation, highlighting both the difficulty of this problem and the generality of CWBB.

\section{Discussion}

This article demonstrates that extending the weighted Bayesian bootstrap to constrained problems yields a simple yet effective method applicable to a diverse range of statistical tasks. Our approach bridges the advantages of uncertainty quantification via approximate posterior sampling to the computational expediency of optimization routines for constrained problems. We identify standard regularity assumptions that lead to new theoretical guarantees and asymptotic results on the samples from Algorithm \ref{algorithm:cwbb}. Just as the bootstrap enjoys widespread usage due to its simplicity---even beyond settings where all regularity conditions are checked---our method is easy to use with any plug-in optimization solution to constrained problems. 

Compared to other constrained posterior sampling methods, CWBB enjoys competitive coverage, runtimes, and sampling efficiency. Our theory shows that relaxation-based methods may not entirely factor in the structure of constraint sets in their uncertainty estimates, while CWBB is centered around an efficient restricted maximum likelihood estimator. As demonstrated by the runtimes in the Supplementary Material, CWBB maintains a smaller computational footprint than several competing methods which require prohibitively expensive evaluation of proximal operators for the log-density and gradient calculations. Although our method is meant to cover cases where dedicated samplers are not available, we note CWBB performs competitively with the bespoke samplers \texttt{BDGraph} and \texttt{ssgraph} as demonstrated in Section \ref{section:bayes-gl}. Furthermore, the options case study of Section \ref{section:option-prices} is particularly difficult, with CWBB as the only method able to produce viable samples quickly from the constrained posterior, without difficulties from mixing or infeasible samples. 

There are several immediate future directions that these contributions imply. First, it is natural to strengthen the theory to encompass a broader class of settings. Indeed, while Theorem~\ref{thm:two} resembles the standard Bernstein--von Mises result, it requires more stringent conditions than some comparable results in the literature (see Supplementary Material). As many of these technical assumptions are imposed in order to produce valid Taylor expansions around the true parameter in both the constraint and likelihood functions, weaker assumptions may apply to broaden the scope of these results with modern posterior analysis techniques such as those from \citet{miller_asymptotic_2021}. Some specific tasks such as high-dimensional regression and normal means models admit attractive theoretical properties under the weighted Bayesian bootstrap with weaker regularity conditions. These include consistency \citep{ng_random_2022} and minimax posterior contraction rates \citep{nieBayesianBootstrapSpikeandSlab2023} under penalty formulations instead of hard constraints, again suggesting that our theory could be made more applicable to specific cases by considering the geometry of the problem.

Our analysis justifies the constrained weighted Bayesian bootstrap in a first-order sense. The support of the prior is considered in Theorem \ref{thm:two}, but we do not fully make use of its functional form. Techniques such as Edgeworth expansions have been used to produce higher-order approximations of the weighted likelihood/Bayesian bootstrap while accounting for prior structure \citep{newtonApproximateBayesianInference1994, pompeIntroducingPriorInformation2021}. Finally, because of the breadth statistical tasks that can be formulated as constrained optimization problems, we invite readers to explore further applications of CWBB to statistical learning settings where constrained spaces pose a challenge to posterior inference. This includes sampling over valid correlation matrices, intersections of convex sets or split problems \citep{xu2016mm}, and combinatorial problems \citep{xu2023bayesian}. 

\bibliography{ref}

\newpage

\onecolumn

\title{Constrained Weighted Bayesian Bootstrap\\(Supplementary Material)}
\maketitle

\appendix
\section{Additional experiment results}

\subsection{Nonnegative Order-Constrained Parameters in Regression}

This experiment is a comprehensive comparison of the constrained posterior sampling methods mentioned in Section~\ref{sec:intro}. In Table~\ref{tab:nonneg-order-runtimes}, we show that CWBB has runtime orders of magnitude faster than several other constrained sampling methods. Most methods performed 250 posterior samples per trial. However, methods requiring HMC were afforded a warmup and 1000 samples (except Proximal Prior with 250 samples) to account for potential correlation across samples. In addition, Fig.~\ref{fig:minimum-coverage} shows the minimum coverage (the lowest point for each dataset size in Fig.~\ref{fig:ordered-coverage}) over all $p=30$ parameters in $\theta$. CWBB is most similar to the baseline unconstrained Gibbs sampler, as compared to the competing methods. Finally, we use the experiments from Section \ref{section:nonnegative-ordered} to empirically demonstrate Theorems \ref{thm:one} and \ref{thm:two}. Figure \ref{fig:theory-validation} shows convergence of CWBB samples to the true underlying parameter and the covariance matrix of the samples to the projected inverse Fisher information.

\begin{table}[!h]
    \centering
    \caption{Time to complete 2500 trials for coverage simulation study.}\label{tab:nonneg-order-runtimes}
    {
    \begin{tabular}{lrr}
        \toprule
		Method & Samples per Trial & Total Runtime (seconds) \\
		\midrule
        Constraint Relaxation \citep{duanBayesianConstraintRelaxation2020} & 1000 & 37958 \\
		Constrained Weighted Bayesian Bootstrap & 250 & 1242 \\
        Transformation of \citet{dunsonBayesianInferenceOrderConstrained2003} & 250 & 2538 \\
        Orthogonally Projected Gibbs Samples & 250 & 255 \\
        Proximal Prior \citep{xuBayesianInferenceUsing2024} & 250 & 83830 \\
        ProxMCMC/Distance-to-set \citep{presmanDistancetoSetPriorsConstrained2023, zhouProximalMCMCBayesian2024} & 1000 &  3200 \\
        \texttt{positive\_ordered} Change of Variables in Stan & 1000 & 3304 \\
        Unconstrained Gibbs Sampler & 250 & 110 \\
        \bottomrule
	\end{tabular}}
\end{table}

\begin{figure}[htb]
\centering
\includegraphics[width=0.58\textwidth]{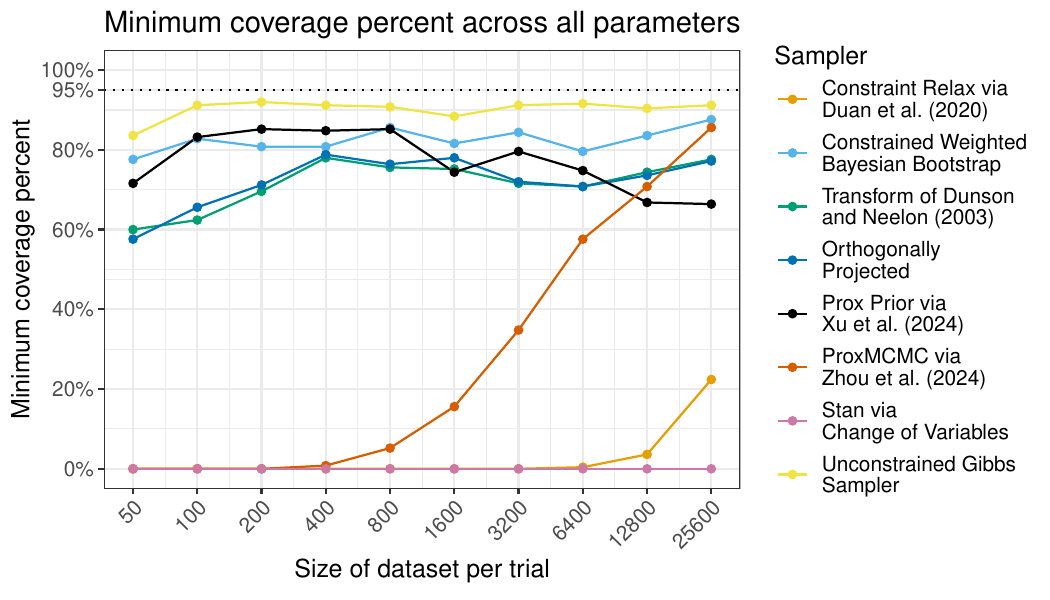}
\caption{For each sampling method, the minimum coverage percent over all $\theta_i$ for each dataset size.}
\label{fig:minimum-coverage}
\end{figure}

\begin{figure}[htb]
\centering
\includegraphics[width=0.50\textwidth]{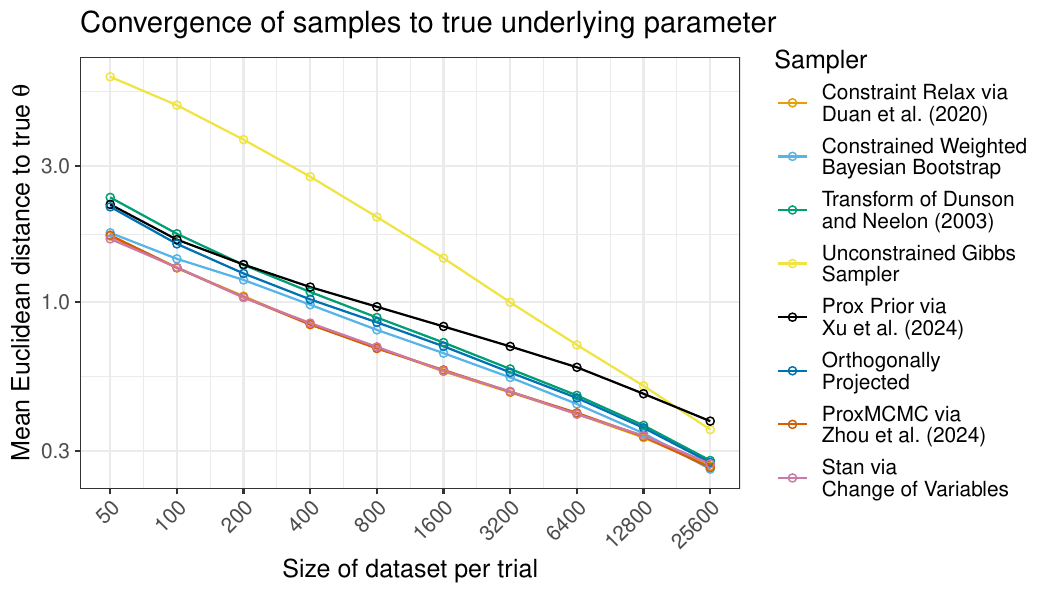}
\includegraphics[width=0.46\textwidth]{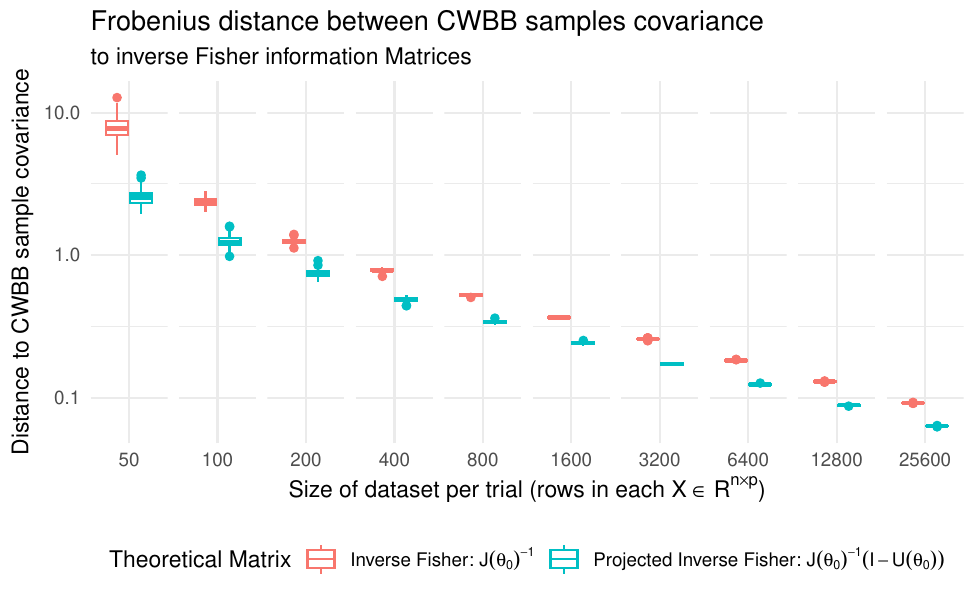}
\caption{(Left) Empirical verification of Theorem \ref{thm:one}, showing CWBB samples converge to the true $\theta$, at a rate similar to that of competing methods. (Right) Empirical verification of Theorem \ref{thm:two}, showing the empirical covariance matrix of CWBB samples is closer in Frobenius norm to the theoretical covariance described in Theorem \ref{thm:two}, as opposed to the theoretical covariance from unconstrained settings.}
\label{fig:theory-validation}
\end{figure}


\newpage
\subsection{Sparse Precision Matrix Estimation}

The first three rows of Table \ref{tab:bayesian-sparse-precision} show the peak performance from Fig.~\ref{fig:graphical-lasso-classification} while the remaining rows show peak performance over a hyperparameter grid for other Bayesian methods. Bayesian methods \texttt{BDgraph} and \texttt{ssgraph} perform 5000 iterations of their respective samplers with the first half discarded as burn-in. ProxMCMC performs 1000 iterations of a No-U-turn sampler \citep{betancourtConceptualIntroductionHamiltonian2018}, along with a dynamic warmup, as implemented in the Julia package \texttt{DynamicHMC.jl}. These methods are afforded more iterations due to potential correlation among successive samples. ProxMCMC is flexible enough to implement the nonconvex penalties SCAD and MCP (see code for implementation details), although performance was not improved with them. The methods of \citet{mohammadiBayesianStructureLearning2015a} and \citet{wangScalingItStochastic2015} have a starting graph where all values of the precision matrix are nonzero; performance was poor otherwise. All $F_1$-scores shown are for the best performing threshold for posterior probability of a precision matrix entry being nonzero. For the first seven rows, this corresponds to the maximum posterior probability of the entry being strictly positive or negative, while the remaining two rows are posterior probabilities of the edge indicator being nonzero. \texttt{R} packages \texttt{BayesianGLasso} \citep{bayesianglasso}, \texttt{BDGraph} \citep{mohammadiBDgraphPackageBayesian2019}, and \texttt{ssgraph} \citep{ssgraph} implement the methods of \citet{wangBayesianGraphicalLasso2012}, \citet{mohammadiBayesianStructureLearning2015a} and \citet{wangScalingItStochastic2015}, respectively. The code to implement ProxMCMC was provided with \citet{zhouProximalMCMCBayesian2024}.

\begin{table}[!h]
    \centering
    \caption{Best performances for Bayesian sparse precision estimates. Runtime is the time to complete samples for the set of given hyperparameters.}\label{tab:bayesian-sparse-precision}
    {
    \begin{tabular}{lrrp{6cm}r}
        \toprule
		Method & $F_1$-score & Decision Rule & Hyperparameters & Runtime (seconds) \\
		\midrule
        CWBB - LASSO & 0.68 & 95\% & $\rho = 2^{-10}$ & 23 \\
		CWBB - SCAD & 0.72 & 75\% & $\rho = 2^{-4.5}, a = 3.7$ & 250  \\
        CWBB - MCP & 0.73 & 65\% & $\rho = 2^{-4}, a = 3$ & 215 \\
        ProxMCMC - LASSO & 0.67 & 95\% & $\lambda = 0.1, \alpha_{\text{scale}} = 400$ & 68 \\
        ProxMCMC - SCAD & 0.66 & 95\% & $\lambda = 10^{-4}, \alpha_{\text{scale}} = 50$ & 294 \\
        ProxMCMC - MCP & 0.67 & 95\% & $\lambda = 10^{-3}, \alpha_{\text{scale}} = 200$ & 145 \\ 
        \texttt{BayesianGLasso} & 0.61 & 85\% & $a = 10^{-4}, b = 2^{14}$ & 343 \\
        \texttt{BDGraph} & 0.76 & 75\% & $\texttt{g.prior} = 0.35$, $\texttt{df.prior} = 3 \times 2^8$, $\texttt{g.start} = ``\text{full}"$ & 78 \\
        \texttt{ssgraph} & 0.78 & 55\% & $\lambda = 2^{-4}$, $\pi = 0.75$, $v_0 = 4\times 10^{-4}$, $v_1 = 1$, $\texttt{g.start} = ``\text{full}"$ & 217 \\       
        \bottomrule
	\end{tabular}}
\end{table}


\subsection{Option Pricing Surfaces}

Examining the left side of Fig.~\ref{fig:option-prices-spread}, the violations explored in Fig.~\ref{fig:option-prices-feasibility} are correlated with a larger spread. Although one might expect options with a wider range of possible prices to have more uncertainty in their underlying price, by factoring in the constraints we are able to incur smaller uncertainty in options with large spread. This is shown in the right side of Fig.~\ref{fig:option-prices-spread} showing many options deep in-the-money (a strike price well below the underlying price of \$2300) have a credible interval width that is a small proportion of the spread, as opposed to those at-the-money which have a width that contains up to half of the spread.

The left side of Fig.~\ref{fig:option-prices-comparison} shows the unconstrained posterior correlation matrix, which contains less rich structure than the constrained one showed in Fig.~\ref{fig:option-prices-correlation-constrained}. There are some similarities, such as weak correlation when examining options that expire two days apart. The right side of Fig.~\ref{fig:option-prices-comparison} shows the difference in point estimates between the constrained and unconstrained posteriors. Deep in-the-money options tend to have a larger difference, although most options have very similar price estimates. However, for a subset of days to expire, there is a noticeable difference in price for some options at-the-money.

\begin{figure}[t]
\centering
\includegraphics[width=0.48\textwidth]{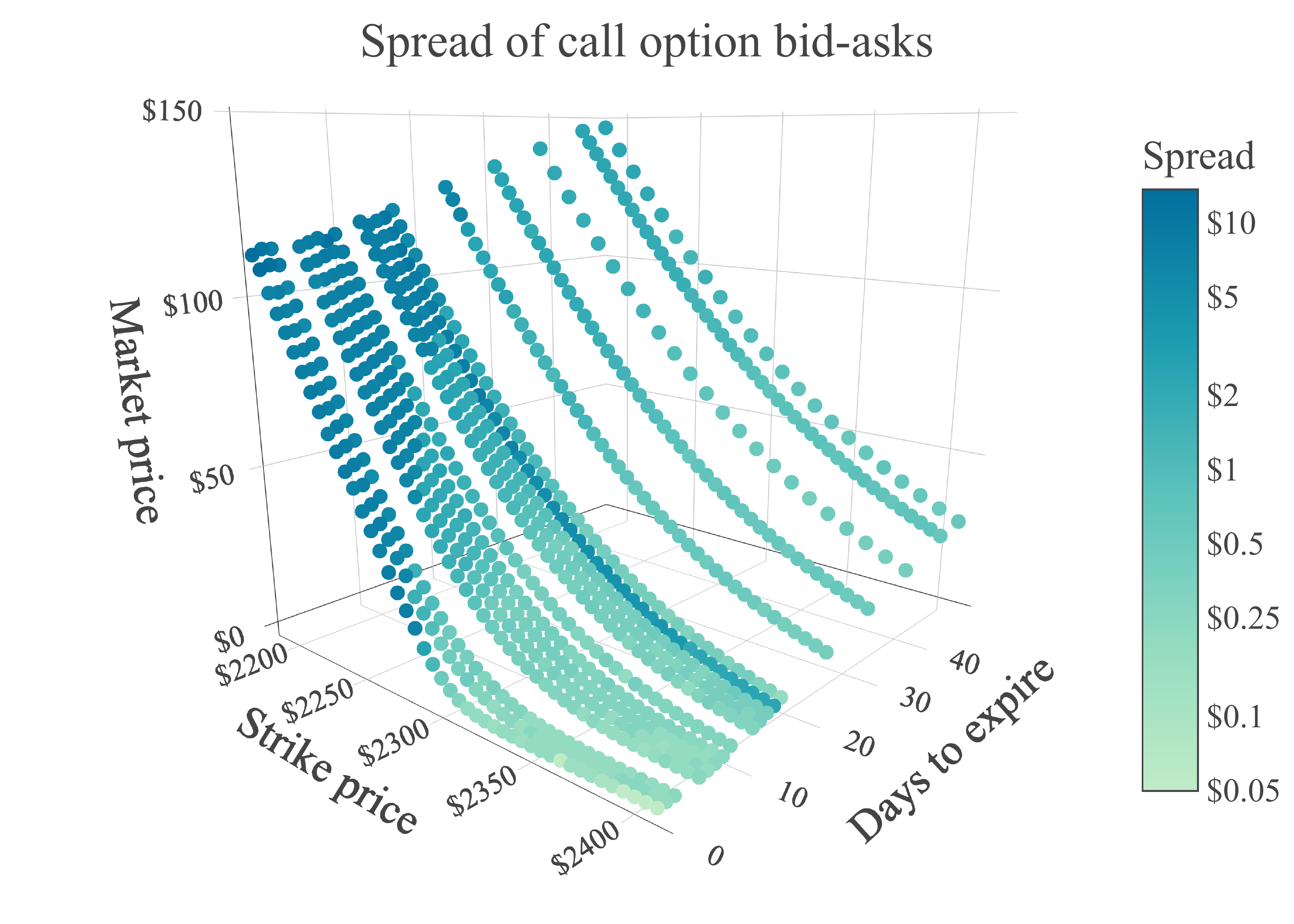}
\includegraphics[width=0.38\textwidth]{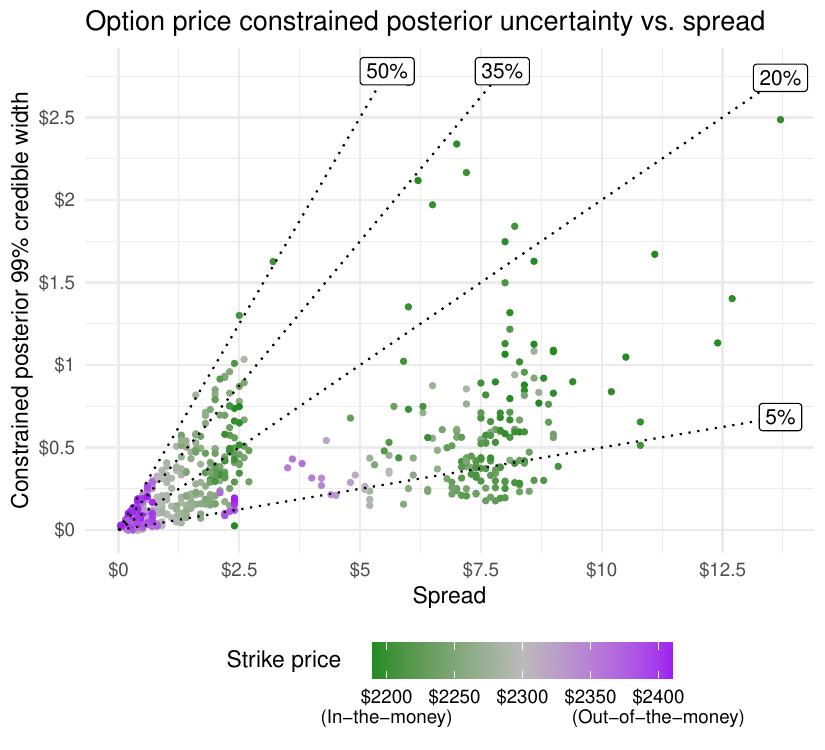}
\caption{(Left) For the dataset of \texttt{RUT} call options, scatter plots in 3D showing the midpoint price (z-axis) as it varies in both strike price (x-axis) and days to expiration (y-axis). Points are colored according to the log of the spread. (Right) Scatter plot showing the spread of options on the x-axis, the posterior credible interval widths on the y-axis, with points colored by underlying strike price. Several references lines are included to show what percent of the spread the credible interval contains.}
\label{fig:option-prices-spread}
\end{figure}

\begin{figure}[!h]
\centering
\includegraphics[width=0.45\textwidth]{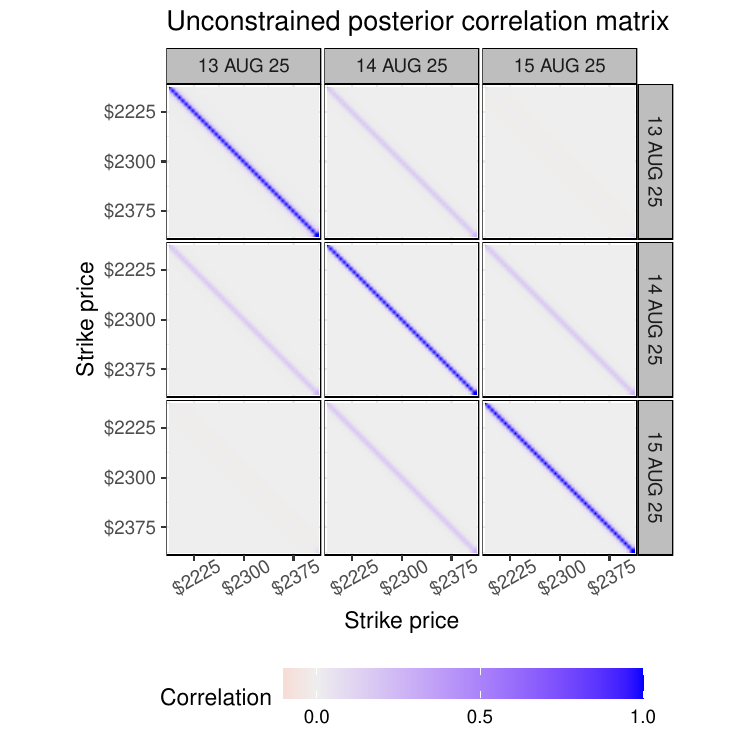}
\includegraphics[width=0.48\textwidth]{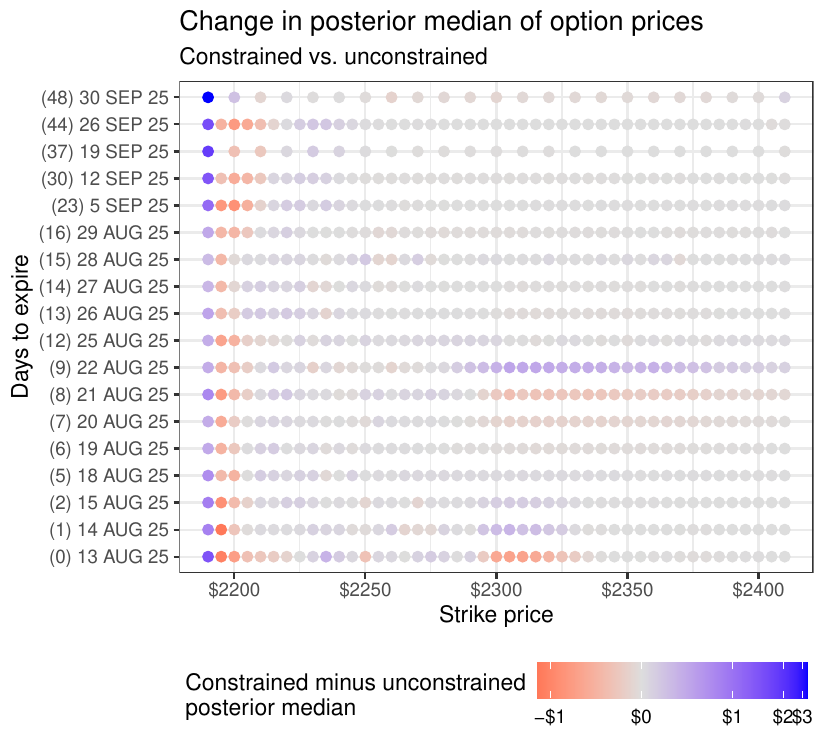}
\caption{(Left) Correlation matrix for unconstrained posterior samples of normal means for the first three times to expire in the dataset. Options are ordered in the matrix by time to expire and then by strike price. (Right) Scatter plot with strike price on the x-axis, days to expire on the y-axis, and points colored by how much larger the constrained posterior median is than the unconstrained.}
\label{fig:option-prices-comparison}
\end{figure}

\paragraph{Comparison with other samplers:} Due to the geometry of the constraint set, modern constrained posterior methods struggle heavily on the options case study, while CWBB is able to produce 1000 independent samples in less than 10 minutes (see Table~\ref{tab:options-runtimes}). Both constraint relaxation \citep{duanBayesianConstraintRelaxation2020} and ProxMCMC \citep{zhouProximalMCMCBayesian2024} (equivalent to distance-to-set priors \citep{presmanDistancetoSetPriorsConstrained2023}) take a significant amount of time longer than CWBB to produce 1000 samples. The samples from ProxMCMC/distance-to-set struggle to mix, giving a mean effective sample size below 5, despite a proper warmup with a No-U-turn-sampler. Both constraint relaxation and ProxMCMC/distance-to-set are unable to produce samples in the constraint set with many samples violating the bid-ask or shape constraints. Attempts to increase the regularization only lead to prohibitive sampling times and difficulty mixing. The proximal prior of \cite{xuBayesianInferenceUsing2024} has extreme computational costs as it requires a numerical approximation of the gradient of the projection onto the constraint polytope. Despite optimized code, it is unable to produce samples in a reasonable amount of time.

\begin{table}[!h]
    \centering
    \caption{Properties of 1000 samples for the case study in Section \ref{section:option-prices}.}\label{tab:options-runtimes}
    {
    \begin{tabular}{lrrr}
        \toprule
		Method & Feasible Samples & Mean Effective Sample Size & Total Runtime (seconds) \\
		\midrule
        Constraint Relaxation & 0 & 603.10 & 60865 \\
		CWBB & 1000 & Gives independent samples & 494 \\
        Proximal Prior & - & - & Failed \\
        ProxMCMC/Distance-to-set & 0 & 4.78 & 6775 \\
        \bottomrule
	\end{tabular}}
\end{table}

\section{Proof of Main Results}
\subsection{Proof Preliminaries}

The techniques used in the resulting proofs are found in Michael Newton's 1991 PhD dissertation at the University of Washington \citep{newtonWeightedLikelihoodBootstrap1991}, and work on restricted maximum likelihood estimation by \citet{aitchisonMaximumLikelihoodEstimationParameters1958}. We combine aspects of notation and regularity conditions from each work and introduce it here. Recall the weights are generated such that
\begin{equation*}
    (w_{n, 1}, \ldots, w_{n,n}) = \frac{n}{\sum_{j=1}^n Y_j}(Y_1, \ldots, Y_n),
\end{equation*}
where $n$ is the sample size, and each $Y_i \sim \text{Exp}(1)$ independently, admitting a $\text{Dirichlet}(1,\ldots,1)$ distribution for $w_{n}/n$. The typical empirical score and information functions, denoting the weighted versions with a tilde, are defined as
\begin{align*}
    S_n(\theta) & = \frac{1}{n}\sum_{i=1}^n \nabla_{\theta} \log f_{\theta}(X_i), &
    \tilde S_n(\theta) = \frac{1}{n}\sum_{i=1}^n w_{n, i} \nabla_{\theta} \log f_{\theta}(X_i), \\
    J_n(\theta) & = -\frac{1}{n} \sum_{i=1}^n \nabla_{\theta}^2 \log f_{\theta}(X_i),
    & \tilde J_n(\theta) = -\frac{1}{n} \sum_{i=1}^n w_{n, i}\nabla_{\theta}^2 \log f_{\theta}(X_i).
\end{align*}
Elements of the Jacobian of the constraint functions and a variant (with expectation always with respect to $\theta_0$) of the Fisher information are defined as
\begin{align*}
    [D_h(\theta)]_{ij} & = \frac{\partial h_i(\theta)}{\partial \theta_j},
    & [J(\theta)]_{ij} = - \mathbb E_{\theta_0}\brk[c]3{\frac{\partial^2 \log f_\theta(X)}{\partial \theta_i \partial \theta_j} }.
\end{align*}
Finally, the third-order partial derivatives of the log-likelihood (see Condition \ref{condition:log-likelihood-partials}) are pertinent for Taylor expansions. Weighted sums of elements of the tensor are denoted
\begin{align*}    
    [\tilde \Psi^j_{1:n}(\theta)]_{k \ell} = \sum_{i=1}^n w_{n,i}\frac{\partial^3 \log f_\theta(X_i)}{\partial \theta_j \partial \theta_k \partial \theta_\ell}.
\end{align*}

Conditional convergence is the end result of the theory. The definition is repeated for clarity. 
\begin{definition}[Convergence in Conditional Probability]
    Let $U$ and $V_1,V_2,\ldots$ be defined on the same probability space. Then $V_n$ converges in conditional probability $a.s.[X_{1:\infty}]$ to $U$ if for all $\epsilon > 0$
    \[
        {\mathrm{pr}}\brk{\|V_n - U\|_2 > \epsilon \mid X_{1:n}} \to 0\quad a.s.[X_{1:\infty}],
    \]
    as $n \to \infty$. In shorthand, this property is denoted $V_n\to_{c.p.} U\ a.s.[X_{1:\infty}]$ and applies for almost every infinite sample sequence.
\end{definition}

There are several shared motifs in the regularity conditions of \citet{newtonWeightedLikelihoodBootstrap1991} and \citet{aitchisonMaximumLikelihoodEstimationParameters1958} as many are used in order to have proper Taylor expansions around the score function. In the following conditions $\theta_0$ refers to the true underlying value of $\theta$.

\begin{condition}[Identifiability]\label{condition:first}\label{condition:identifiability}
For any $\theta_1 \neq \theta_0$, there exists a set $\mathcal D$ where $P_{\theta_1}(\mathcal D) \neq P_{\theta_0}(\mathcal D)$ and $P_\theta$ is the probability measure induced from probability density function $f_{\theta}$. In addition, 
\begin{equation*}
    \mathbb E_{\theta_0}\brk[c]3{\frac{\log f_{\theta_0}(x)}{\log f_{\theta_1}(x)}} < \infty.
\end{equation*}
\end{condition}

\begin{condition}[Feasibility]\label{condition:feasibility}

The support set, $\Theta$, has a nonempty subset $\tilde \Theta = \brk[c]{\theta \colon h(\theta) = 0}$ and $\theta_0 \in \tilde \Theta$.
\end{condition}

\begin{condition}[Neighborhood of $\theta_0$]\label{condition:neighborhood}
There is a closed ball $U_{\alpha} = \{\theta \colon \|\theta - \theta_0\|_2 \leq \alpha\} \subset \Theta$ around $\theta_0$ for some $\alpha$.
\end{condition}

\begin{condition}[Smoothness of unconstrained prior]\label{condition:prior-smoothness}
For every $\theta \in U_\alpha$, the function $\log \pi(\theta)$ is bounded and continuous. Furthermore, for every $\theta \in U_\alpha$, the following derivatives exist, are bounded, and are continuous,
\begin{equation*}
    \frac{\partial \log \pi(\theta)}{\partial \theta_j}, \quad (j=1, \ldots, p).
\end{equation*}
\end{condition}

\begin{condition}[Smoothness of log-likelihood] \label{condition:log-likelihood-partials}
For every $\theta \in U_\alpha$, the following derivatives exist for almost all $x$ and are continuous with respect to $\theta$,
    \begin{equation*}
        \frac{\partial \log f_\theta(x)}{\partial \theta_j}, \quad  
        \frac{\partial^2 \log f_\theta(x)}{\partial \theta_j \partial \theta_k}, \quad  
        \frac{\partial^3 \log f_\theta(x)}{\partial \theta_j \partial \theta_k\partial \theta_\ell},\quad (j,k,\ell=1,\ldots, p).
    \end{equation*}
\end{condition}

\begin{condition}[Boundedness of likelihood]\label{condition:density-boundedness}
For every $\theta \in U_\alpha$, for all $(j,k,\ell=1,\ldots, p)$, there exist functions $m^{(\cdot)}$ such that
    \begin{align}
        \bigg| \frac{\partial f_\theta(x)}{\partial \theta_j} \bigg| & < m_j^{(1)}(x), \quad \int m_j^{(1)}(x) \mu(dx) < \infty, \nonumber \\
        \bigg| \frac{\partial^2 f_\theta(x)}{\partial \theta_j \partial \theta_k} \bigg| & < m_{jk}^{(2)}(x), \quad \int m_{jk}^{(2)}(x) \mu(dx) < \infty, \nonumber \\
        \bigg| \frac{\partial^2 \log  f_\theta(x)}{\partial \theta_j \partial \theta_k} \bigg| & < m_{jk}^{(3)}(x), \quad  \mathbb E_{\theta_0}\brk[c]{ m_{jk}^{(3)}(X)} < \infty, \nonumber \\
        \bigg|\frac{\partial^3 \log f_{\theta}(x)}{\partial \theta_j \partial \theta_k \partial \theta_\ell}\bigg| & < m^{(4)}_{jk\ell}(x), \quad \mathbb E_{\theta_0}\brk[c]{m^{(4)}_{jk\ell}(X)} < \infty, \label{eqn:assumption-g4} \\
        \bigg|\frac{\partial \log f_{\theta}(x)}{\partial \theta_j}\cdot \frac{\partial^2 \log f_{\theta}(x)}{\partial \theta_k \partial \theta_\ell}\bigg| & < m^{(5)}_{jk\ell}(x), \quad \mathbb E_{\theta_0}\brk[c]{m^{(5)}_{jk\ell}(X)} < \infty. \nonumber
    \end{align}

\end{condition}

\begin{condition}[Positive-definite information matrix]\label{condition:pos-def-information}
The Fisher information $J(\theta_0)$ is positive-definite with finite elements. Furthermore, $I(\theta) = \mathbb E_{\theta}\brk[s]{\brk[c]{\nabla_{\theta} \log f_{\theta}(X)}\brk[c]{ \nabla_{\theta} \log f_{\theta}(X)}^\top}$ is positive-definite with finite elements for all for all $\theta \in U_\alpha$. Note that $I(\theta_0) = J(\theta_0)$ under Conditions \ref{condition:log-likelihood-partials} and \ref{condition:density-boundedness} \citep{newtonWeightedLikelihoodBootstrap1991}.
\end{condition}

\begin{condition}[Smoothness of equality constraint function]\label{condition:constraint-smoothness}
For every $\theta \in U_\alpha$, the following derivatives exist and are continuous functions of $\theta$, 
    \begin{equation*}
        \frac{\partial h_k(\theta)}{\partial \theta_j},\quad (j=1,\ldots, p;\ k=1,\ldots, r).
    \end{equation*}
    In addition, for every $\theta \in U_\alpha$, the following derivatives exist, are continuous functions of $\theta$ and are bounded above by some constant $c_h$,
    \begin{equation*}
        \bigg|\frac{\partial^2 h_k(\theta)}{\partial \theta_j\partial \theta_\ell}\bigg| \leq c_h, \quad (j,\ell=1,\ldots, p;\ k=1,\ldots, r). 
    \end{equation*}
    
\end{condition}

\begin{condition}[Full-rank Jacobian at $\theta_0$] \label{condition:full-rank-jacobian}\label{condition:last}
    The $r \times p$ matrix $D_h(\theta_0)$ has full row-rank.
\end{condition}

We quickly prove under the conditions that the covariance matrix from Theorem \ref{thm:two} is invariant to $h$.

\begin{lemma}\label{lemma:invariant-h}
    Suppose $g(\theta)$ and $h(\theta)$ satisfy Conditions \ref{condition:feasibility}, \ref{condition:constraint-smoothness} and \ref{condition:full-rank-jacobian}. If
    \begin{equation*}
        U_h(\theta) = D_h^\top(\theta) \brk[c]{D_h(\theta_0) J(\theta_0)^{-1} D_h^\top(\theta)}^{-1} D_h(\theta_0) J(\theta_0)^{-1},
    \end{equation*}
    then $U_h(\theta_0) = U_g(\theta_0)$.
\end{lemma}

\begin{proof}
    Let $\mathcal F_h = \{D_h(\theta_0) d = 0\}$, be the set of feasible directions at $\theta_0$, under the definition of $h$ (Definition 12.3 of \citet{nocedalTheoryConstrainedOptimization2006}). Under Condition \ref{condition:full-rank-jacobian} and Lemma 12.2 of \cite{nocedalTheoryConstrainedOptimization2006}, $\mathcal F_h$ is equal to the tangent space of $\tilde \Theta$ at $\theta_0$. The same is true of $\mathcal F_g$. Hence, $D_h(\theta_0)$ and $D_g(\theta_0)$ have same kernel and row space. Hence, there exists some invertible matrix $P \in \mathbb R^{r \times r}$ such that $D_h(\theta_0) = P D_g(\theta_0)$. Then
    \begin{align*}
        U_h(\theta_0) & = D_h^\top(\theta_0) \brk[c]{D_h(\theta_0) J(\theta_0)^{-1} D_h^\top(\theta_0)}^{-1} D_h(\theta_0) J(\theta_0)^{-1}, \\
        & = D_g^\top(\theta_0) P^\top \brk[s]{P \brk[c]{D_g(\theta_0) J(\theta_0)^{-1} D_g^\top (\theta_0)} P^\top }^{-1} P D_g(\theta_0) J(\theta_0)^{-1}, \\
        & = D_g^\top(\theta_0) P^\top (P^\top)^{-1}{ \brk[c]{D_g(\theta_0) J(\theta_0)^{-1} D_g^\top (\theta_0)}  }^{-1} P^{-1} P D_g(\theta_0) J(\theta_0)^{-1}, \\
        & = U_g(\theta_0).
    \end{align*}
\end{proof}

Without proof we repeat the following useful lemmas from \citet{newtonWeightedLikelihoodBootstrap1991}.

\begin{lemma}[Lemma 3 of \citet{newtonWeightedLikelihoodBootstrap1991}] \label{lemma:weighted-avg-convergence}
    Let the data $X_i$ be drawn from density $f_{\theta_0}$ and let the random variables $Y_i$ be independent, identically distributed from a unit exponential, and independent of $X$. If $m$ is a real-valued, measurable function such that $\mathbb E_{\theta_0}\brk[c]{m(X_i)} < \infty$, then
    \begin{equation*}
        \frac{1}{n}\sum_{i=1}^n Y_i m(X_i) \to_{c.p.} \mathbb E_{\theta_0}\brk[c]{m(X_i)}\quad a.s.[X_{1:\infty}].
    \end{equation*}
\end{lemma}

\begin{lemma}[Lemmas 4 and 5 of \citet{newtonWeightedLikelihoodBootstrap1991}] \label{lemma:weighted-score-information-convergence}
    Under Conditions \ref{condition:first}, \ref{condition:neighborhood}, \ref{condition:log-likelihood-partials}, and \ref{condition:density-boundedness}, as $n \to \infty$,
    \begin{align*}
        \|\tilde S_n(\theta_0)\|_2 & \to_{c.p.} 0 \quad a.s.[X_{1:\infty}], \\
        \| \tilde J_n(\theta) - J(\theta) \|_2 & \to_{c.p.}0 \quad a.s.[X_{1:\infty}].
    \end{align*}
\end{lemma}

\begin{lemma}[Lemma 8 of \citet{newtonWeightedLikelihoodBootstrap1991}] \label{lemma:unit-vector-convergence}
    Let $z \in \mathbb R^{p}$ be a unit vector. Let $a_{in} = z^\top \nabla_{\theta} \brk[c]{\log f_{\hat \theta_n}(X_i)}$, for a strongly consistent estimator, $\hat \theta_n$. Under Conditions \ref{condition:identifiability}, \ref{condition:neighborhood}, \ref{condition:log-likelihood-partials}, \ref{condition:density-boundedness}, and \ref{condition:pos-def-information}, as $n\to\infty$,
    \begin{align*}
        \frac{1}{n} \sum_{i=1}^n a_{in}^2 & \to z^\top J(\theta_0) z \quad a.s.[X_{1:\infty}], \\
        \frac{1}{n}\max_{i=1,\ldots,n} a_{in}^2 & \to 0 \quad a.s.[X_{1:\infty}].
    \end{align*}
\end{lemma}

\begin{lemma}[Theorem 15 of \citet{newtonWeightedLikelihoodBootstrap1991}] \label{lemma:lindeberg-clt-corollary}
    Let $Z_1, Z_2,\dots$ be independent and identically distributed random variables with mean $\mu$ and variance $\sigma^2$. Let $(a_{in})$ be a non-vanishing sequence of constants satisfying 
    \begin{equation*}
        \frac{\sum_{i=1}^n a_{in}^2}{\max_{i=1,\ldots,n} a_{in}^2} \to \infty, \quad n\to\infty.
    \end{equation*} Then for $T_n = \sum_{i=1}^n a_{in} Z_i,\ \mu_n = \mu \sum_{i=1}^n a_{in},$ and $\sigma^2_n = \sigma^2 \sum_{i=1}^n a_{in}^2$,
    \begin{equation*}
        \frac{T_n - \mu_n}{\sigma_n}
    \end{equation*}
    converges to a standard normal distribution as $n\to\infty$.
\end{lemma}

\subsection{Proof of Proposition \ref{proposition:penalized-posterior-distribution}}

\begin{proof}
    First, the conditions in the proof statement are made rigorous by Condition 11 from \cite{miller_asymptotic_2021} which we assume is satisfied.
    Under Theorem 12 of \cite{miller_asymptotic_2021}, with probability 1, $h_n(\theta) = -n^{-1} \sum_{i=1}^n \log f_\theta(x_i)$ satisfies the conditions of Theorem 5 of \cite{miller_asymptotic_2021}. We now show that $g_n(\theta) = h_n(\theta) + \rho q(\theta)$ also satisfies these conditions with probability 1, which implies the proposition statement.

    Continuity of $g_n$ at $\theta_0$ is followed by the assumed continuity of $q$ at $\theta_0$. In addition, there is some $E$, with $\theta_0 \in E$, which is open and bounded where $h_n$ has continuous third derivatives which are uniformly bounded. Let $B$ be the interior of the intersection of $E$ and the closed ball around $\theta_0$ where $q(\theta)$ is three-times continuously differentiable. Then, $B$ is an open set containing $\theta_0$ where $g_n$ has continuous third derivatives. Furthermore, $g_n$ has uniformly bounded third derivatives on the closure of $B$, as it is compact. This implies uniformly bounded third derivatives on $B$ itself. Trivially, as $h_n$ converges pointwise to some function $h$ (established by Theorem 5 of \cite{miller_asymptotic_2021}), then $g_n$ converges pointwise to $g$ where $g(\theta) = h(\theta) + \rho q(\theta)$. We also have $\nabla^2_{\theta} h(\theta_0) = -\mathbb E_{\theta_0}\brk[c]{\nabla^2_\theta \log f_{\theta_0}(X)}$ which is positive-definite via properties of natural exponential families. Hence, under assumption, $\nabla^2_\theta g(\theta_0)$ is also positive-definite.

    Finally we show Condition (1) of Theorem 5 of \citet{miller_asymptotic_2021}. Apply Theorem 12 of \citet{miller_asymptotic_2021} which says with probability one, for any open ball $G$ with $\theta_0 \in G$ and $\operatorname{closure}(G) \subset \Theta$ that there exists some compact set $K$ with $\theta_0$ in the interior where $\liminf_n \inf_{\theta \in \Theta \setminus K} h_n(\theta) > h(\theta_0)$. Hence,
    \begin{align}
        g(\theta_0) = h(\theta_0) < \liminf_n \inf_{\theta \in \Theta \setminus K} h_n\brk{\theta} \leq \liminf_n \inf_{\theta \in \Theta \setminus K} h_n\brk{\theta} + \rho q(\theta) = \liminf_n \inf_{\theta \in \Theta \setminus K} g_n(\theta), 
    \end{align}
    via $q(\theta) \geq 0$ and $q(\theta_0) = 0$. Also on $K$ we have $h(\theta) > h(\theta_0)$ for all $\theta \in K\setminus\{\theta_0\}$. Because $q(\theta_0) = 0$ and $q(\theta) \geq 0$, this further implies $g(\theta) > g(\theta_0)$ for all $\theta \in K\setminus\{\theta_0\}$.
\end{proof}

\subsection{Proof of Theorem \ref{thm:one}}

The following two lemmas of \citet{aitchisonMaximumLikelihoodEstimationParameters1958} help establish the existence of a constrained weighted Bayesian bootstrap solution. The first is equivalent to Brouwer's fixed-point theorem, and the second is reproved for conditional probabilities.

\begin{lemma}[Lemma 2 of \citet{aitchisonMaximumLikelihoodEstimationParameters1958}] \label{lemma:fixed-point-optima}
If $m$ is a continuous function mapping a closed ball to itself, such that for every $\|z\|_2 = 1$, $z^\top m(z) < 0$, then there exists a point $\hat z$ such that $\|\hat z\|_2 < 1$ and $m(z) = 0$.    
\end{lemma}

\begin{lemma}[Analogue to Lemma 1 of \citet{aitchisonMaximumLikelihoodEstimationParameters1958}]\label{lemma:remove-multipliers}
    Consider the Lagrangian multiplier system of the constrained weighted Bayesian bootstrap with no inequality constraints,
    \begin{align}
        0 & = n \tilde S_n(\theta) + \nabla_\theta \log \pi (\theta) +  D_h^\top(\theta) \lambda, \label{eqn:weighted-system-lagrangian} \\ 
        0 & = h(\theta). \label{eqn:weighted-system-feasibility}
    \end{align}
    Subject to Conditions \ref{condition:first}--\ref{condition:full-rank-jacobian}, for $\delta < \alpha$ and almost all $w_n$ we have that 
    $\check \theta_n \in U_{\delta}$ and $\check \lambda_n$ satisfy \eqref{eqn:weighted-system-lagrangian} and \eqref{eqn:weighted-system-feasibility} if and only if $\check \theta_n$ satisfies an equation of the form
    \begin{equation}
        -J(\theta_0) (\check \theta_n - \theta_0) + \tilde v(\check \theta_n) = 0. \label{eqn:critical-equation}
    \end{equation}
    Furthermore, $\tilde v$ is continuous with respect to $\theta$ on $U_\delta$ for almost all $X_{1:n}$ and $w_n$. In addition, there exists a constant $c_v$ (that does not depend on $w_n, X, \delta$, or $n$) such that 
    \begin{equation*}
        {\mathrm{pr}}\brk[c]3{\sup_{\theta \in U_{\delta}} \|\tilde v(\theta) \|_2 \leq \delta^2 c_{v} \mid X_{1:n}} \to 1\quad a.s.[X_{1:\infty}].
    \end{equation*}
\end{lemma}

\begin{proof}
    Let $\theta \in U_\delta$ with $U_\delta$ well defined from Conditions \ref{condition:first}--\ref{condition:neighborhood}. Under Conditions \ref{condition:feasibility}, \ref{condition:log-likelihood-partials}, and \ref{condition:constraint-smoothness}, perform a Taylor expansion of the weighted score function and the constraint function around $\theta_0$ and define $\tilde v^{(1)}$ and $v^{(2)}$ such that
    \begin{align}
        n \tilde S_n(\theta) & = n \tilde S_n(\theta_0) - n \tilde J_n(\theta_0)(\theta - \theta_0) + \tilde v^{(1)}(\theta), \label{eqn:weighted-score-expansion} \\
        h(\theta) & = D_h(\theta_0)(\theta - \theta_0) + v^{(2)}(\theta). \label{eqn:constraint-expansion}
    \end{align}
    Components of $\tilde v^{(1)}(\theta) \in \mathbb R^p$ and $v^{(2)}(\theta) \in \mathbb R^r$ can be written as
    \begin{align}
        \tilde v^{(1)}_j(\theta) & = \frac{1}{2}(\theta - \theta_0)^\top \tilde \Psi_{1:n}^j(\theta^{(1,*)}) (\theta - \theta_0), \label{eqn:v1-components} \\
        v^{(2)}_j(\theta) & = \frac{1}{2}(\theta - \theta_0)^\top \nabla_{\theta}^2 h_j(\theta^{(2, *)}) (\theta - \theta_0), \label{eqn:v2-components}
    \end{align}
    where $\theta^{(\cdot, *)}$ are some vectors on the line segment between $\theta$ and $\theta_0$. Substitute the expansions \eqref{eqn:weighted-score-expansion} and \eqref{eqn:constraint-expansion} into the Lagrangian system \eqref{eqn:weighted-system-lagrangian} and \eqref{eqn:weighted-system-feasibility}
    \begin{align}
        n \tilde S_n(\theta_0) - n \tilde J_n(\theta_0) ( \theta - \theta_0) + D_h^\top(\theta) \lambda + \nabla_{\theta}\log \pi(\theta) + \tilde v^{(1)}(\theta) & = 0, \label{eqn:thm-expansion-1} \\
        D_h(\theta_0) (\theta - \theta_0) + v^{(2)}(\theta) & = 0. \label{eqn:thm-expansion-2}
    \end{align}
     Now assume $\check \theta_n \in U_\delta$ and $\check \lambda_n$ satisfy \eqref{eqn:weighted-system-lagrangian} and \eqref{eqn:weighted-system-feasibility}, substitute $n \tilde S_n(\check \theta_n) = -D_h^\top(\check \theta_n) \lambda - \nabla_{\theta}\log \pi(\check \theta_n)$ into \eqref{eqn:thm-expansion-1}, $h(\check \theta_n)$ into \eqref{eqn:thm-expansion-2} and then simplify:
    \begin{align}
        n \tilde S_n(\theta_0) - n \tilde J_n(\theta_0) ( \check \theta_n - \theta_0) + \tilde v^{(1)}(\check \theta_n) & = n\tilde S_n(\check \theta_n), \label{eqn:intermediate-1} \\
        D_h(\theta_0) (\check \theta_n - \theta_0) + v^{(2)}(\check \theta_n) & = h(\check \theta_n) = 0. \label{eqn:intermediate-2}
    \end{align}
    Divide \eqref{eqn:intermediate-1} by $n$ and define $\tilde v^{(3)}(\theta)$ as 
    \begin{align}
        -J(\theta_0)(\check \theta_n - \theta_0) + \tilde v^{(3)}(\check \theta_n) & = \tilde S_n(\check \theta_n) = -\frac{D_h^\top(\check \theta_n) \check \lambda_n}{n} - \frac{\nabla_\theta \log \pi(\check \theta_n)}{n}, \label{eqn:intermediate-3} \\
        \tilde v^{(3)}(\theta) & = \tilde S_n(\theta_0) - \brk[c]{\tilde J_n(\theta_0) - J(\theta_0)}(\theta - \theta_0) + \frac{\tilde v^{(1)}(\theta)}{n}. \label{eqn:v3-def}
    \end{align}
    Via the definition of $\tilde v^{(1)}(\theta)$ we also have
    \begin{equation}
        \tilde v^{(3)}(\theta) = \tilde S_n(\theta) + J(\theta_0)(\theta - \theta_0). \label{eqn:v3-sol-def}
    \end{equation}
    Multiply both sides of \eqref{eqn:intermediate-3} by $D_h(\theta_0) J(\theta_0)^{-1}$ and substitute $D_h(\theta_0) (\check \theta_n - \theta_0)$ from \eqref{eqn:intermediate-2} to get
    \begin{equation*}
        v^{(2)}(\check \theta_n) + D_h(\theta_0) J(\theta_0)^{-1} \brk[c]3{ \tilde v^{(3)}(\check \theta_n) + \frac{\nabla_\theta \log \pi(\check \theta_n)}{n}} = - D_h(\theta_0) J(\theta_0)^{-1}D_h^\top(\check \theta_n)\frac{ \check \lambda_n}{n}.
    \end{equation*}
    This gives an expression for $\check \lambda_n$
    \begin{equation*}
        \check \lambda_n = -n \brk[c]{D_h(\theta_0) J(\theta_0)^{-1} D_h^\top(\check \theta_n)}^{-1}\brk[s]3{D_h(\theta_0) J(\theta_0)^{-1} \brk[c]3{\tilde v^{(3)}(\check \theta_n) + \frac{\nabla_\theta \log \pi(\check \theta_n)}{n}}+ v^{(2)}(\check \theta_n)  }.
    \end{equation*}
    Substitute $\check \lambda_n$ into \eqref{eqn:intermediate-3} to get
    \begin{equation*}
        -J(\theta_0)(\check \theta_n - \theta_0) + \tilde v(\check \theta_n) = 0,
    \end{equation*}
    where
    \begin{equation}
        \begin{split}
            \tilde v(\theta) & = -D_h^\top(\theta)\brk[c]{D_h(\theta_0) J(\theta_0)^{-1} D_h^\top(\theta)}^{-1} \brk[s]3{D_h(\theta_0) J(\theta_0)^{-1} \brk[c]3{\tilde v^{(3)}(\theta) + \frac{\nabla_\theta \log \pi(\theta)}{n}}+ v^{(2)}(\theta)  } \\
            & \quad + \tilde v^{(3)}(\theta) + \frac{\nabla_\theta \log \pi(\theta)}{n}. \label{eqn:v-def}
        \end{split}
    \end{equation}
    Using $\eqref{eqn:weighted-system-lagrangian}$ for $\tilde S_n(\check \theta_n) = -D_h^\top(\check \theta_n)\check \lambda_n / n - \nabla_\theta \log \pi(\check \theta_n)/ n $, \eqref{eqn:weighted-system-feasibility} for $h(\check \theta_n) = 0$, and $\tilde v^{(3)}(\check \theta_n)$ from \eqref{eqn:v3-sol-def}, it is clear $-J(\theta_0) (\check \theta_n - \theta_0) + \tilde v(\check \theta_n) = 0$. 

    Suppose now that $-J(\theta_0)(\theta - \theta_0) + \tilde v(\theta) = 0$ and $\theta \in U_\delta$. Multiply this expression by $D_h(\theta_0) J(\theta_0)^{-1}$ to get
    \begin{equation*}
        \begin{split}
        0 & = D_h(\theta_0) J(\theta_0)^{-1} \brk[c]{\tilde v(\theta) -J(\theta_0)(\theta - \theta_0)} \\
        & = -D_h^\top(\theta_0) J(\theta_0)^{-1} \brk[c]3{\tilde v^{(3)}(\theta) + \frac{\nabla_\theta \log \pi (\theta) }{n}} - v^{(2)}(\theta) \\
        & \quad + D_h(\theta_0) J(\theta_0)^{-1} \brk[c]3{\tilde v^{(3)}(\theta) + \frac{\nabla_\theta \log \pi (\theta)}{n}} - D_h(\theta_0) (\theta - \theta_0)
        \end{split}
    \end{equation*}
    Substituting \eqref{eqn:constraint-expansion} implies $h(\theta) = 0$. Using $h(\theta) = 0$ and \eqref{eqn:v3-sol-def} we can do similar substitution as above for
    \begin{equation*}
       \tilde S_n(\theta) + \frac{\nabla_\theta \log \pi(\theta)}{n} = D_h^\top(\theta)\underbrace{\brk[c]{D_h(\theta_0)J(\theta_0)^{-1} D_h^\top(\theta)}^{-1}D_h(\theta_0) J(\theta_0)^{-1}\brk[c]3{\tilde S_n(\theta) + \frac{\nabla_\theta \log \pi(\theta)}{n} } }_{-\lambda / n}.
    \end{equation*}
    This then implies
        $n\tilde S_n(\theta) + \nabla_\theta \log \pi(\theta) + D_h^\top(\theta) \lambda = 0$, for some $\lambda$ due to \eqref{eqn:critical-equation} being satisfied.

    We now show for small $\delta$ that $\tilde v(\theta)$ is bounded and continuous for $\theta \in U_\delta$ with high conditional probability. We reiterate definitions for $\tilde v^{(1)}, v^{(2)}$, $\tilde v^{(3)}$ and show each is bounded. We can easily bound the components of $v^{(2)}$ defined in \eqref{eqn:v2-components} due to Condition \ref{condition:constraint-smoothness}:
    \begin{align}
        v^{(2)}_j(\theta) & = \frac{1}{2}(\theta - \theta_0)^\top \nabla_{\theta}^2  h_j(\theta^{(2, *)}) (\theta - \theta_0), \nonumber \\
        & \leq \sup_{\theta' \in U_\delta} \frac{1}{2} |\lambda_{\max}\brk[c]{\nabla_\theta^2 h_j(\theta')}| \|\theta' - \theta\|_2^2, \nonumber \\
        & \leq \frac{\delta^2 p c_h}{2}. \label{eqn:gershgorin-bound}
    \end{align}
    Line \eqref{eqn:gershgorin-bound} is due to the maximum eigenvalue being bounded by the magnitude of the entries times the number of rows from the Gershgorin circle theorem. This holds almost surely with respect to the random weights. Similarly for $\tilde v^{(1)}(\theta)$ defined in \eqref{eqn:v1-components}:
    \begin{align*}
        |\tilde v^{(1)}_j(\theta)| & = \frac{1}{2} \bigg| \sum_{k=1}^p \sum_{\ell=1}^p [\theta - \theta_0]_k [\theta - \theta_0]_\ell \brk[c]3{\sum_{i=1}^n w_{n,i}\frac{\partial^3 \log f_{\theta^*}(X_i)}{\partial \theta_j \partial \theta_k \partial \theta_\ell}} \bigg|, \\
        & \leq \frac{1}{2} \|\theta - \theta_0\|_\infty^2 \sum_{k=1}^p \sum_{\ell=1}^p {\sum_{i=1}^n \bigg| w_{n,i}\frac{\partial^3 \log f_{\theta^*}(X_i)}{\partial \theta_j \partial \theta_k \partial \theta_\ell}} \bigg|, \\
        & \leq \frac{\delta^2}{2} \sum_{k=1}^p \sum_{\ell=1}^p {\sum_{i=1}^n w_{n,i} m_{jk\ell}^{(4)}(X_i)}, \\
        & = \frac{\delta^2}{2\bar Y} \sum_{k=1}^p\sum_{\ell=1}^p \frac{1}{n} \brk[c]3{\sum_{i=1}^n Y_i m_{jk\ell}^{(4)}(X_i)} \to_{c.p.} \frac{\delta^2 p^2 c_4}{2} \quad a.s.[X_{1:\infty}].
    \end{align*}
    Convergence in probability is due to \eqref{eqn:assumption-g4} in Condition \ref{condition:density-boundedness} and Lemma \ref{lemma:weighted-avg-convergence}. The above implies
    \begin{align*}
        {\mathrm{pr}}\brk[c]3{\|\tilde v^{(1)}(\theta)\|_\infty > {\delta^2p^2c_4} \mid X_{1:n}} & = {\mathrm{pr}}\brk[c]3{\bigcup_{j=1}^p |\tilde v^{(1)}_j(\theta)| > {\delta^2p^2c_4} \mid X_{1:n}}, \\
        & \leq \sum_{j=1}^p {\mathrm{pr}}\brk[c]3{ |\tilde v^{(1)}_j(\theta)| > {\delta^2p^2c_4} \mid X_{1:n}}, \\
        & \leq \sum_{j=1}^p {\mathrm{pr}}\brk[c]3{\frac{\delta^2}{2\bar Y} \sum_{k=1}^p\sum_{\ell=1}^p \frac{1}{n} \brk[c]3{\sum_{i=1}^n Y_i m_{jk\ell}^{(4)}(X_i)} > \delta^2 p^2 c_4 \mid X_{1:n}}, \\
        & \to 0 \quad a.s.[X_{1:\infty}].
    \end{align*}
    Using the explicit definition of $\tilde v^{(3)}(\theta)$ from \eqref{eqn:v3-def}, Cauchy-Schwarz, triangle inequality, and norm equivalence imply
    \begin{equation*}
        \|\tilde v^{(3)}(\theta)\|_2 \leq \|\tilde S_n(\theta_0)\|_2 + \|\tilde J_n(\theta_0) - J(\theta_0)\|_2\|\theta - \theta_0\|_2 + \frac{p^{1/2}}{n}\|\tilde v^{(1)}(\theta)\|_\infty.
    \end{equation*}
    Via Lemma \ref{lemma:weighted-score-information-convergence} we can choose $n$ large enough so ${\mathrm{pr}}\brk[c]{\|\tilde S_n(\theta_0)\|_2 \geq \delta^2 \mid X_{1:n}} < \eta / 3$, ${\mathrm{pr}}\brk[c]{\|\tilde J_n(\theta_0) - J(\theta_0)\|_2 \geq \delta \mid X_{1:n}} < \eta/3 $, and ${\mathrm{pr}}\brk[c]{\|\tilde v^{(1)}(\theta)\|_\infty > \delta^2 p^2 c_4 \mid X_{1:n}} < \eta/3 $, for $\eta > 0$. This then implies via a standard union bound
    \begin{equation*}
        {\mathrm{pr}}\brk[c]3{\|\tilde v^{(3)}(\theta)\|_2 > \delta^2\brk3{2 + \frac{p^{5/2}c_4}{2n}} \mid X_{1:n}} \leq \eta. \label{eqn:v3-bound}
    \end{equation*}
    
    We examine the terms of \eqref{eqn:v-def} to continue. We see $\nabla_\theta \log \pi$ is bounded and defined in $U_\delta$ via Conditions \ref{condition:neighborhood} and \ref{condition:prior-smoothness}. Choosing $n$ large enough will ensure $\nabla \log \pi(\theta)/n$ will have entries below $\delta^2$. Bounding the operator norms of the matrices in $\eqref{eqn:v-def}$ will further bound $\tilde v$. The operator norms of $D_h(\theta_0)$ and $J(\theta_0)$ are constants and hence bounded. Because $D_h$ is continuous it will have bounded entries within the ball $U_\alpha$. Since $\delta < \alpha$, the operator norm of $D_h(\theta)$ is uniformly bounded in $U_\delta$. By choosing $\alpha$ small enough, $D_h(\theta)$ will also be full rank for all $\theta \in U_\alpha$ ensuring $\brk[c]{D_h(\theta_0) J(\theta_0)^{-1} D_h^\top(\theta)}^{-1}$ is nonsingular, has bounded entries, and bounded operator norm. All of these operator norms will factor into $c_v$. Given all components not containing $w_n$ in \eqref{eqn:v-def} can be bounded below a factor of $\delta^2$ for all $\theta \in U_\delta$, combined with \eqref{eqn:v3-bound} the $\ell_2$-norm of the $\tilde v$ function for $\theta \in U_\delta$ is bounded by a factor of $\delta^2$ with arbitrarily high conditional probability as $n$ grows.

    Continuity follows from similar arguments. Via \eqref{eqn:constraint-expansion}, $v^{(2)}$ is continuous in $U_\delta$, since $h$ is also continuous in this neighborhood (Condition \ref{condition:constraint-smoothness}). Each component of $\tilde v^{(1)}(\theta)$ is continuous due to the continuity of the third partial derivatives of $\log f_\theta$ for all $\theta \in U_\delta$ (Condition \ref{condition:log-likelihood-partials}). Continuity of $\tilde v^{(3)}$ follows from continuity of $\tilde v^{(1)}$. Both $D_h$ and $\nabla_{\theta} \log \pi$ are continuous by the same conditions that have boundedness. This leaves showing $\brk[c]{D_h(\theta_0) J(\theta_0)^{-1} D_h^\top( \theta)}^{-1}$ is continuous on $U_\delta$. We have that there is some $\delta$ where $D_h(\theta)$ is full row-rank for all $\theta \in U_{\delta}$, since $D_h(\theta_0)$ is full row-rank and $D_h$ is continuous (Conditions \ref{condition:constraint-smoothness}, \ref{condition:full-rank-jacobian}), all perturbations below a certain magnitude do not change the rank. As a result, $D_h(\theta_0) J(\theta_0)^{-1} D_h^\top( \theta)$ is full-rank and bounded in $U_\delta$, for $\delta$ small enough, so the inverse is also continuous. Finally, all terms in the definition of $\tilde v$ on $U_\delta$ are continuous so it is as well. This holds for almost all $w_n$.
\end{proof}

We now prove Theorem 1.

\begin{proof}
    We emulate the proof of Theorem 1 from \citet{aitchisonMaximumLikelihoodEstimationParameters1958}, while taking additional care for the conditional probabilities. Let $c_{v}$ be defined in Lemma \ref{lemma:remove-multipliers}. Fix $\delta < \min\brk[s]{\lambda_{\min}\brk[c]{J(\theta_0)} / c_{v}, c_{v}^{-1/2} \lambda_{\max}^{1/2}\brk[c]{J(\theta_0)}, \alpha, 1}$. Let $\theta$ be such that $\|\theta - \theta_0\|_2 = \delta$. Define on the closed unit ball in $\mathbb R^p$,
    \begin{equation}
        m\brk3{\frac{\theta - \theta_0}{\delta}} = \frac{1}{2 \lambda_{\max}\brk[c]{J(\theta_0)}}\brk[c]{-J(\theta_0)(\theta - \theta_0) + \tilde v(\theta)}. \label{eqn:m-def}
    \end{equation}
    Via Lemma \ref{lemma:remove-multipliers}, $m$ is continuous for all $\theta \in U_{\delta}$ and almost all $X_{1:n}, w_n$. In addition, for all $\theta \in U_{\delta}$,
    \begin{equation*}
        \bigg \|m\brk3{\frac{\theta - \theta_0}{\delta}} \bigg\|_2 \leq \frac{\|\theta - \theta_0\|_2}{2} + \frac{1}{2\lambda_{\max}\brk[c]{J(\theta_0)}}\|\tilde v(\theta)\|_2.
    \end{equation*}
    If $\|\tilde v(\theta)\|_2 < \lambda_{\max}\brk[c]{J(\theta_0)}$, then $\|m\|_2 \leq 1$ for all $\theta \in U_{\delta}$. Let $\epsilon < 1$, so via Lemma \ref{lemma:remove-multipliers} there exists $n(\delta, \epsilon)$ large enough such that
    \begin{equation*}
        {\mathrm{pr}}\brk[c]3{\sup_{\theta \in U_{\delta}} \|\tilde v(\theta) \|_2 \leq \delta^2 c_{v} \mid X_{1:n}} > 1 - \epsilon / 2, \quad (n=n(\delta, \epsilon), \ldots, \infty).
    \end{equation*}
    Then for $\delta <  {c_{v}^{-1/2} \lambda_{\max}^{1/2}\brk[c]{J(\theta_0)}}$, we have $\|\tilde v(\theta)\|_2 < \lambda_{\max}\brk[c]{J(\theta_0)}$ with conditional probability at least $1-\epsilon/2$. Hence, all $m$ maps its domain to itself with conditional probability at least $1-\epsilon/2$ for $n > n(\delta, \epsilon)$.
    
    We now show that Lemma \ref{lemma:fixed-point-optima} can be applied with high conditional probability. With probability at least $1-\epsilon/2$, for $\theta$ with $\|\theta - \theta_0\|_2 = \delta$,
    \begin{align}
        \frac{1}{\delta} (\theta - \theta_0)^\top m\brk3{\frac{\theta - \theta_0}{\delta}} & = \frac{1}{2\delta \lambda_{\max}\brk[c]{J(\theta_0)}}\brk[c]{-(\theta - \theta_0)^\top J(\theta_0)(\theta - \theta_0) + (\theta - \theta_0)^\top\tilde v(\theta)}, \nonumber \\
        & \leq \frac{1}{2\delta \lambda_{\max}\brk[c]{J(\theta_0)}}\brk[s]{-\lambda_{\min}\brk[c]{J(\theta_0)}\|\theta - \theta_0\|_2^2 + \|\theta - \theta_0\|_2 \|\tilde v(\theta)\|_2}, \nonumber \\ 
        & \leq \frac{1}{2\delta \lambda_{\max}\brk[c]{J(\theta_0)}}\brk[s]{-\lambda_{\min}\brk[c]{J(\theta_0)}\|\theta - \theta_0\|_2^2 + \delta^2 c_{v} \|\theta - \theta_0\|_2}, \nonumber \\
        & = \frac{1}{2\lambda_{\max}\brk[c]{J(\theta_0)}}\brk[s]{-\delta \lambda_{\min}\brk[c]{J(\theta_0)} + \delta^2 c_{v}}. \label{eqn:thm2-negative-terms}
    \end{align}
    Line \eqref{eqn:thm2-negative-terms} is negative when $\delta < {\lambda_{\min}\brk[c]{J(\theta_0)} / c_{v}}$. As a result, we can apply Lemma \ref{lemma:fixed-point-optima} and see there must exist a point $\check \theta_{n}$ where $\check \theta_{n} \in U_\delta$, but not on the boundary, and $m\brk[c]{(\check \theta_{n} - \theta_0) / \delta} = 0$. \eqref{eqn:m-def} is satisfied by $\check \theta_{n}$ so via Lemma \ref{lemma:remove-multipliers}, $\check \theta_{n}$ also satisfies \eqref{eqn:weighted-system-lagrangian} and \eqref{eqn:weighted-system-feasibility} and this holds with arbitrarily high conditional probability as $n$ grows. Furthermore, the Lagrange multipliers exist from Lemma \ref{lemma:remove-multipliers}.
    We immediately see $\check \theta_{n}$ is conditionally consistent because there exists $n(\delta, \epsilon)$ for arbitrarily small $\delta$ and $\epsilon$ such that both qualifiers for Lemma \ref{lemma:fixed-point-optima} are satisfied and
    \begin{equation}
        {\mathrm{pr}}\brk{\|\check \theta_n - \theta_0\|_2 < \delta \mid X_{1:n}} > 1 - \epsilon,
    \end{equation}
    for all $n > n(\delta, \epsilon)$.
\end{proof}
\subsection{Proof of Theorem \ref{thm:two}}

We start with an additional lemma to assist in the proof.

\begin{lemma} \label{lemma:score-cancellation}
    Let $U(\theta) = D_h^\top(\theta) \brk[c]{D_h(\theta_0) J(\theta_0)^{-1}D_h^\top(\theta)}^{-1} D_h(\theta_0) J(\theta_0)^{-1}$ be defined under Conditions \ref{condition:pos-def-information} and \ref{condition:full-rank-jacobian}. Let $\hat \theta_n$ be a strongly consistent estimator for $\theta_0$ such that $h(\hat \theta_n) = 0$ almost surely and
    \begin{equation}
        \|n^{1/2} \brk[c]{I - U(\hat \theta_n)} S_n(\hat \theta_n)\| \to 0 \quad a.s.[X_{1:\infty}]. \label{eqn:score-projection-convergence}
    \end{equation}
    Furthermore, let $\tilde M_n$ be such that
    \begin{equation}
        \tilde M_n \to_{c.p.} J(\theta_0)^{-1}\brk[c]{I - U(\theta_0)} \quad a.s.[X_{1:\infty}]. \label{eqn:M-prob-convergence}
    \end{equation}
    Then
    \begin{equation*}
        n^{1/2} \tilde M_n \tilde S_n(\hat \theta_n) \mid X_{1:n} \Rightarrow N\brk[s]{0, J(\theta_0)^{-1}\brk[c]{I - U(\theta_0)}}\quad a.s.[X_{1:\infty}].
    \end{equation*}
\end{lemma}

\begin{proof}
    Let $t_n(z) = n^{1/2}z^\top \tilde S_n(\hat \theta_n)$  and $a_{in}(z) = z^\top \nabla_{\theta} \brk[c]{\log f_{\hat \theta_n}(X_i)}$. Rearranging terms we have
    \begin{align*}
        t_n(z) & = n^{1/2} \sum_{j=1}^p z_j \brk3{\frac{\sum_{i=1}^n Y_i \frac{\partial }{\partial \theta_j} \log f_{\hat \theta_n}(X_i)}{\sum_{i=1}^n Y_i}} = \frac{1}{n^{1/2} \bar Y} \sum_{i=1}^n Y_i a_{in}(z).
    \end{align*}
    The strong law of large numbers implies $\bar Y$ converges almost surely (with respect to the weights probability measure) to one, so $t_n(z)$ will converge in distribution similarly to $n^{-1/2} \sum_{i=1}^n Y_i a_{in}(z)$.
    Under Lemma \ref{lemma:unit-vector-convergence} we have 
    \begin{equation*}
        \frac{1}{n}\sum_{i=1}^n a_{in}^2(z) \to z^\top J(\theta_0) z, \quad \frac{1}{n}\max_{i=1,\ldots,n} a_{in}^2(z) \to 0 \quad a.s.[X_{1:\infty}],
    \end{equation*}
    so the conditions to apply Lemma \ref{lemma:lindeberg-clt-corollary} hold $a.s.[X_{1:\infty}]$ on this quantity and as a result,
    \begin{align*}
        t_n(z) - n^{1/2} z^\top S_n(\hat \theta_n) \mid X_{1:n} \Rightarrow N\brk[c]{0, z^\top J(\theta_0) z} \quad a.s.[X_{1:\infty}].
    \end{align*}
    Applying the Cram\'er-Wold theorem,
    \begin{equation*}
        n^{1/2} \brk[c]{\tilde S_n(\hat \theta_n) - S_n(\hat \theta_n)} \mid X_{1:n} \Rightarrow N\brk[c]{0, J(\theta_0)} \quad a.s.[X_{1:\infty}].
    \end{equation*}
    
    Multiplying both sides by $\tilde M_n$ and applying a conditional version of Slutsky's lemma (along almost every sample path) with \eqref{eqn:M-prob-convergence} gives 
    \begin{equation*}
        n^{1/2} \tilde M_n \brk[c]{\tilde S_n(\hat \theta_n) - S_n(\hat \theta_n)} \mid X_{1:n} \Rightarrow N\brk[s]{0, J(\theta_0)^{-1} \brk[c]{I - U(\theta_0)}} \quad a.s.[X_{1:\infty}].
    \end{equation*}
    To complete the proof we show the asymptotic distribution of $n^{1/2} \tilde M_n S_n(\hat \theta_n) \mid X_{1:n}$ is a point mass at zero. First, by \eqref{eqn:M-prob-convergence}, using conditional Slutsky's, for almost every sample path,
    \begin{equation*}
        n^{1/2} \tilde M_n S_n(\hat \theta_n) \mid X_{1:n} \to_{c.p.} J(\theta_0)^{-1} \brk[c]{I - U(\theta_0)} S_n(\hat \theta_n)\quad a.s.[X_{1:\infty}].
    \end{equation*}
    As a consequence of \eqref{eqn:score-projection-convergence}, $n^{1/2}\brk[c]{I - U(\theta_0)} S_n(\hat \theta_n)$ is trivially conditionally consistent for 0. As a result, for almost every sample path, $n^{1/2} \tilde M_n S_n(\hat \theta_n)$ is conditionally consistent to 0 and the proof statement holds.
\end{proof}

We now prove the full theorem.

\begin{proof}
    Using Condition \ref{condition:constraint-smoothness}, assume without loss of generality, that $\alpha$ is small enough such that $D_h(\theta)$ is full-rank with linearly independent rows for all $\theta \in U_\alpha$. Then Conditions \ref{condition:feasibility} and \ref{condition:full-rank-jacobian} apply and a necessary condition for $\check \theta_n$ to be a local optima (as a sample from Algorithm \ref{algorithm:cwbb}) is satisfying the system
    \begin{align}
        \begin{split} \label{eqn:thm2-intermediate-1}
            0 & = \tilde S_n(\check \theta_n) + \frac{\nabla_{\theta} \log \pi (\check \theta_n)}{n} + \frac{D_h^\top(\check \theta_n) \check \lambda_n}{n}, \\
            0 & = h(\check \theta_n). \\
        \end{split}
    \end{align}
    Via Theorem \ref{thm:one} we have shown that $\check \theta_n$ is conditionally consistent for $\theta_0$, so we can let $n$ be large enough that $\check \theta_n \in U_\alpha$ with conditional probability at least $1-\epsilon$. 
    
    Using the strong consistency of $\hat \theta_n$ we can let $n$ be large enough so Condition \ref{condition:log-likelihood-partials} applies almost surely, giving a Taylor expansion around $\hat \theta_n$ of $\tilde S_n(\theta)$ to get
    \begin{equation}
        \tilde S_n(\theta) = \tilde S_n(\hat \theta_n) - \tilde J_n(\hat \theta_n) (\theta - \hat \theta_n) + r(\theta), \label{eqn:thm2-proof-expanded}
    \end{equation}
    where $r(\theta)$ has $j$th entry $\frac{1}{2}(\theta - \hat \theta_n)^\top \tilde \Psi_{1:n}^j(\theta_S^*) (\theta - \hat \theta_n)$ such that $\theta_S^*$ is some point on the line segment between $\theta$ and $\hat \theta_n$. Equation \eqref{eqn:thm2-proof-expanded} can then be rearranged to 
    \begin{equation}
        \tilde S_n(\hat \theta_n) = \tilde S_n(\theta) + \brk[c]{\tilde J_n(\hat \theta_n) - R_n(\theta)} (\theta - \hat \theta_n), \label{eqn:thm2-intermediate-2}
    \end{equation}
    where $R_n \in \mathbb R^{p\times p}$ such that 
    \begin{equation*}
    \operatorname{row}_j\brk[c]{R_n(\theta)} = \frac{1}{2}(\theta - \hat \theta_n)^\top \tilde \Psi_{1:n}^j(\theta_S^*). 
    \end{equation*}
    Substitute the stationary condition of \eqref{eqn:thm2-intermediate-1} into \eqref{eqn:thm2-intermediate-2} evaluated at $\check \theta_n$ for
    \begin{equation}
        \tilde S_n(\hat \theta_n) = -\frac{\nabla_\theta \log \pi(\check \theta_n)}{n} - \frac{D_h^\top(\check \theta_n) \check \lambda_n}{n} + \brk[c]{\tilde J_n(\hat \theta_n) - R_n(\check \theta_n)} (\check \theta_n - \hat \theta_n). \label{eqn:thm2-system-1}
    \end{equation}

    By definition $h(\hat \theta_n) = 0$. From Condition \ref{condition:constraint-smoothness} we can expand $h$ around $\hat \theta_n$ and putting $\check \theta_n$ into the expansion gives
    \begin{align}
        0 & = \brk[c]{D_h^\top(\hat \theta_n) + Q_h(\check \theta_n)}^\top(\check \theta_n - \hat \theta_n), \label{eqn:thm2-system-2} \\
        \operatorname{row}_j\brk[c]{Q_h(\check \theta_n)} & = (\check \theta_n - \hat \theta_n)^\top \nabla_\theta^2 h_j(\theta_h^*), \nonumber
    \end{align}
    where $\theta_h^*$ is some point on the line segment between $\check \theta_n$ and $\hat \theta_n$.

    Let 
    \begin{equation*}
        M_n(\hat \theta_n, \check \theta_n) = \begin{bmatrix}
            \tilde J_n(\hat \theta_n) - R_n(\check \theta_n) & - D_h^\top(\check \theta_n) \\
            -\brk[c]{D_h^\top(\hat \theta_n) + Q_h(\check \theta_n)}^\top & 0
        \end{bmatrix}.
    \end{equation*}
    In the proof of Theorem 7 of \citet{newtonWeightedLikelihoodBootstrap1991}, from Conditions \ref{condition:identifiability}, \ref{condition:neighborhood}, \ref{condition:log-likelihood-partials}, \ref{condition:density-boundedness}, \ref{condition:pos-def-information}, and the strong consistency of $\hat \theta_n$, they show $\tilde J_n(\hat \theta_n) - R_n(\check \theta_n) \to_{c.p.} J(\theta_0)\ a.s.[X_{1:\infty}]$. Due to the continuity of $D_h, Q_h$ (Condition \ref{condition:constraint-smoothness}) and the conditional consistency of $\check \theta_n$, the continuous mapping theorem implies $D_h(\check \theta_n) \to_{c.p.} D_h(\theta_0)\ a.s.[X_{1:\infty}]$ and $D_h(\hat \theta_n) + Q_h(\check \theta_n) \to_{c.p.} D_h(\theta_0)\ a.s.[X_{1:\infty}]$. From these conditional consistency results, we have
    \begin{equation*}
        M_n(\hat \theta_n, \check \theta_n) \to_{c.p.} \begin{bmatrix}
            J(\theta_0) & - D_\star^\top \\
            -D_\star & 0 
    \end{bmatrix} \quad a.s.[X_{1:\infty}],
    \end{equation*}
    where $D_\star^\top =  D_h^\top (\theta_0)$. The matrix $M_n$ converges to is nonsingular as $J(\theta_0)$ and $D_h(\theta_0)$ are full-rank.
    
    We combine \eqref{eqn:thm2-system-1} and \eqref{eqn:thm2-system-2} and multiply by $n^{1/2}$ to get the following linear system
    \begin{equation*}
        \begin{bmatrix}
            n^{1/2} \tilde S_n(\hat \theta_n) + n^{-1/2}\nabla_\theta \log \pi(\check \theta_n) \\
            0
        \end{bmatrix} = M_n(\hat \theta_n, \check \theta_n) \begin{bmatrix}
            {n}^{1/2}(\check \theta_n - \hat \theta_n) \\
            n^{-1/2} \check \lambda_n
        \end{bmatrix}.
    \end{equation*}
    On a set with arbitrarily high conditional probability we may write 
    \begin{equation*}
        M_n^{-1}(\hat \theta_n, \check \theta_n)\begin{bmatrix}
            n^{1/2} \tilde S_n(\hat \theta_n) + n^{-1/2} {\nabla_\theta \log \pi(\check \theta_n)} \\
            0
        \end{bmatrix} = \begin{bmatrix}
            n^{1/2}(\check \theta_n - \hat \theta_n) \\
            n^{-1/2}{\check \lambda_n}
        \end{bmatrix}.
    \end{equation*}
    The matrix inverse is continuous in neighborhoods of full-rank, and all functions in $M_n$ are continuous, so we can apply the continuous mapping theorem for
    \begin{equation}
        M_n^{-1}(\hat \theta_n, \check \theta_n) \to_{c.p.} \begin{bmatrix}
                J(\theta_0) & - D_\star^\top \\
                -D_\star & 0 
        \end{bmatrix}^{-1} = \begin{bmatrix}
            P & Q \\ Q^\top & R
        \end{bmatrix}, \label{eqn:block-inv-converge}
    \end{equation}
    where, using formulas for the inverse of block matrices \citep{matrixcookbook},
    \begin{align*}
        P & = J(\theta_0)^{-1}\brk{I - U}, \\
        Q & = -J(\theta_0)^{-1}D_\star^\top \brk[c]{D_\star J(\theta_0)^{-1}D_\star^\top}^{-1}, \\
        R & = -\brk[c]{D_\star J(\theta_0)^{-1}D_\star^\top}^{-1}, \\
        U & = D_\star^\top \brk[c]{D_\star J(\theta_0)^{-1}D_\star^\top}^{-1} D_\star J(\theta_0)^{-1}.
    \end{align*}
    Let $\tilde M_n$ be the top-left block of $M_n^{-1}(\hat \theta_n, \check \theta_n)$, so
    $n^{1/2} (\check \theta_n - \hat \theta_n) = n^{1/2} \tilde M_n \brk[c]{\tilde S_n(\hat \theta_n) + n^{-1/2}\nabla_\theta \log \pi(\check \theta_n) } $.
    Applying Lemma \ref{lemma:score-cancellation} with \eqref{eqn:block-inv-converge} and noting $n^{-1/2} \nabla_{\theta} \log \pi(\check \theta_n)$ converges in conditional distribution to zero (from Condition \ref{condition:prior-smoothness}), we have the proof statement.
\end{proof}

\end{document}